\documentclass[acmsmall,screen, nonacm]{acmart}
\AtBeginDocument{%
  }
\usepackage[utf8]{inputenc}
\usepackage[ruled]{algorithm}
\usepackage[noend]{algpseudocode}
\usepackage{amsmath}
\usepackage{amsfonts}
\usepackage{mathtools}
\usepackage[disable]{todonotes}
\usepackage[ligature,reserved]{semantic}
\usepackage{xspace}
\usepackage{csquotes}

\usepackage[switch]{lineno}

\usepackage{thmtools, thm-restate}
\usepackage{scalerel}
\usepackage{tikz}
\usepackage{adjustbox}

\usetikzlibrary{svg.path}



 \definecolor{orcidlogocol}{HTML}{A6CE39}
\tikzset{
    orcidlogo/.pic={
        \fill[orcidlogocol] svg{M256,128c0,70.7-57.3,128-128,128C57.3,256,0,198.7,0,128C0,57.3,57.3,0,128,0C198.7,0,256,57.3,256,128z};
        \fill[white] svg{M86.3,186.2H70.9V79.1h15.4v48.4V186.2z}
        svg{M108.9,79.1h41.6c39.6,0,57,28.3,57,53.6c0,27.5-21.5,53.6-56.8,53.6h-41.8V79.1z M124.3,172.4h24.5c34.9,0,42.9-26.5,42.9-39.7c0-21.5-13.7-39.7-43.7-39.7h-23.7V172.4z}
        svg{M88.7,56.8c0,5.5-4.5,10.1-10.1,10.1c-5.6,0-10.1-4.6-10.1-10.1c0-5.6,4.5-10.1,10.1-10.1C84.2,46.7,88.7,51.3,88.7,56.8z};
    }
}
\newcommand\orcidicon[1]{\href{https://orcid.org/#1}{\mbox{\scalerel*{
                \begin{tikzpicture}[yscale=-1,transform shape]
                \pic{orcidlogo};
                \end{tikzpicture}
            }{|}}}}

\algblockdefx[Upon]{Upon}{EndUpon}%
[1]{{\bf upon} (#1) {\bf do}}%
{{\bf end upon}}
\algblockdefx[When]{When}{EndWhen}%
[1]{{\bf when} (#1) {\bf do}}%
{{\bf end when}}
\algblockdefx[Receive]{Receive}{EndReceive}%
[1]{{\bf receive} (#1) {\bf do}}%
{{\bf end receive}}
\makeatletter
\ifthenelse{\equal{\ALG@noend}{t}}%
  {\algtext*{EndUpon}}
  {}%
\makeatother

\makeatletter
\ifthenelse{\equal{\ALG@noend}{t}}%
  {\algtext*{EndWhen}}
  {}%
\makeatother

\makeatletter
\ifthenelse{\equal{\ALG@noend}{t}}%
  {\algtext*{EndReceive}}
  {}%
\makeatother

\newcommand\martin[1]{\todo[backgroundcolor=green!23]{Martín: {#1}}}

\newcommand\marga[1]{\todo{Marga: {#1}}}
\newcommand\imarga[1]{\todo[inline]{Marga: {#1}}}

\newtheorem{property}{Property}
\newtheorem{corollary}{Corollary}
\newtheorem{lemma}{Lemma}
\newtheorem{proposition}{Proposition}
\newtheorem{theorem}{Theorem}

\usepackage{xspace}
\usepackage{url}
\usepackage{paralist}

\usepackage{environ}
\usepackage{wrapfig}

\usepackage{bibunits}
\usepackage{hyperref}

\newcommand{\hash}{\mathcal{H}}
\newcommand{\Mroot}{\mathsf{Mroot}}

\newcommand{\repeattheorem}[1]{%
  \begingroup
  \renewcommand{\thetheorem}{\ref{#1}}%
  \expandafter\expandafter\expandafter\theorem
  \csname reptheorem@#1\endcsname
  \endtheorem
  \endgroup
}

\NewEnviron{reptheorem}[1]{%
  \global\expandafter\xdef\csname reptheorem@#1\endcsname{%
    \unexpanded\expandafter{\BODY}%
  }%
  \expandafter\theorem\BODY\unskip\label{#1}\endtheorem
}

\NewEnviron{replemma}[1]{%
  \global\expandafter\xdef\csname replemma@#1\endcsname{%
    \unexpanded\expandafter{\BODY}%
  }%
  \expandafter\lemma\BODY\unskip\label{#1}\endlemma
}
\newcommand{\repeatatlemma}[1]{%
  \begingroup
  \renewcommand{\thelemma}{\ref{#1}}%
  \expandafter\expandafter\expandafter\lemma
  \csname replemma@#1\endcsname
  \endlemma
  \endgroup
}

\NewEnviron{repproposition}[1]{%
  \global\expandafter\xdef\csname repproposition@#1\endcsname{%
    \unexpanded\expandafter{\BODY}%
  }%
  \expandafter\proposition\BODY\unskip\label{#1}\endproposition
}
\newcommand{\repeatatproposition}[1]{%
  \begingroup
  \renewcommand{\theproposition}{\arabic{proposition}}%
  \expandafter\expandafter\expandafter\proposition
  \csname repprop@#1\endcsname
  \endproposition
  \endgroup
}

\newcommand{\PROP}[1]{\textbf{\textit{#1}}\xspace}
\newcommand{\PrAvailability}{\PROP{Availability}}
\newcommand{\PrTermination}{\PROP{Termination}}
\newcommand{\PrValidity}{\PROP{Validity}}

\newcommand{\PROPsub}[2]{\ensuremath{\bf \textit{\textbf{#1}}_{#2}}\xspace}
\newcommand{\PrIntegrityOne}{\PROPsub{Integrity}{1}}
\newcommand{\PrIntegrityTwo}{\PROPsub{Integrity}{2}}
\newcommand{\PrUniqueBatch}{\PROP{Unique Batch}}
\newcommand{\PrCertified}{\PROP{Certified}}
\newcommand{\PrLegality}{\PROP{Legality}}

\newcommand{\Ctt}{\texttt{C}\xspace}

\newcommand{\Rtt}{\texttt{R}\xspace}
\newcommand{\Att}{\texttt{A}\xspace}
\newcommand{\Btt}{\texttt{B}\xspace}
\newcommand{\Ptt}{\texttt{P}\xspace}

\newcommand{\KWD}[1]{\ensuremath{\mathit{#1}}\xspace}
\newcommand{\CC}[1]{\ensuremath{cc_{\KWD{#1}}}\xspace}
\newcommand{\SC}[1]{\ensuremath{sc_{\KWD{#1}}}\xspace}
\newcommand{\SR}[1]{\ensuremath{sr_{\KWD{#1}}}\xspace}
\newcommand{\CR}[1]{\ensuremath{cr_{\KWD{#1}}}\xspace}

\newcommand{\ADV}[1]{#1\xspace}
\newcommand{\advOne}{\ADV{BFT}}
\newcommand{\advTwo}{\ADV{DAC}}
\newcommand{\advThree}{\ADV{Arranger}}

\newcommand{\s}{\ensuremath{S}\xspace}

\reservestyle{\fpm}{\text}
\fpm{fpm[FP mechanism],fpms[FP mechanisms]}
\fpm{fp[FP],fps[FPs]}

\mathlig{<-}{\leftarrow}
\mathlig{==}{\equiv}
\mathlig{++}{\mathop{+\!\!+}}
\mathlig{>>}{\rangle}
\mathlig{~}{\sim}
\reservestyle{\variables}{\text}
\variables{epoch,current,history,theset[the\_set],proposal,valid}

\reservestyle{\setops}{\textsc}
\setops{add,get,Init,Add,Get,BAdd,EpochInc,Broadcast,Deliver,GetEpoch,Propose,Inform,SetDeliver,
  SuperSetDeliver, DoStuff[Start],Consensus}
\setops{translate, Translate}
\setops{SelectSubpath}
\reservestyle{\structs}{\text}
\structs{DPO,BAB,BRB,SBC,DSO}
\structs{ArrangerApi[Arranger], SequencerApi[Sequencer], DCsApi[DC
  members], DCApi[DC member]}

\reservestyle{\stmt}{\textbf}
\stmt{call,ack,drop,return,assert,wait, post, inform, trigger, logger}

\reservestyle{\mathfunc}{\mathsf}
\mathfunc{send,receive,sendbrb[send\_brb],sendsbc[send\_sbc],
  BLSAggregate[AggregateSignatures], sign, broadcast, deliver}

\reservestyle{\messages}{\texttt}
\messages{madd[add],mepochinc[epinc]}

\reservestyle{\api}{\texttt}
\api{apiBAdd,apiAdd,apiGet,apiEpochInc,apiTheSet[{the\_set}],apiHistory}

\reservestyle{\servers}{\texttt}
\servers{S,C,B,M,DAC}

\reservestyle{\schain}{\texttt}
\schain{history,epoch,theset[{the\_set}],add,get,epochinc[{epoch\_inc}],getepoch[{get\_epoch}],tobroadcast[{to\_broadcast}], knowledge, pending, sent, received,proposal,Valid}
\schain{add, translate, allTxs, pendingTxs, maxTime, maxSize,
  signReq, signResp, hashes, setchain, localsetchain[Setchain],
  invalidId, invalidHash, notContacted[left], signatures, getHashAt,
  SubpathChallenged, bisectSubpath,revealSibling, SubpathBisected, init,
  initMultistepMembership, initIntegrity1, initIntegrity2, initDataAvailability, initDecompressAndHash,
  selectSubpath, selectPath, getHashInMiddlePath,
  decompressAndHash,postCompressed, compressedBatch, decompressAndHash} 

\reservestyle{\vars}{\textit}
\vars{prop,propset,h,tx,id,sigs,hash, batchId[id], mt, root,
levels, leaves, top, bottom, middle, pathlength[path\_length],
bottomlevel[bottom\_level], middlelevel[middle\_level],
indexatlevel[index\_at\_level], siblingindex[sibling\_index],
expectedmiddlehash[expected\_middle\_hash], true, false,
bitatpos[bitAtPosition], bt, data, initState, finalState, midState, data}

\reservestyle{\event}{\texttt}
\event{myturn[my\_turn], newepoch[new\_epoch],added,
  timetopost[time\_to\_post],tobatch[to\_batch],
  nextBatch[next\_batch\_id]}

\newcommand{\bottomlevel}{\textit{bottom\_level}}
\newcommand{\pathlength}{\textit{path\_length}}

\newcommand{\inlinefrugal}[1]{{\lstinline[language=Frugal,mathescape]{#1}}}

\usepackage{listings, xcolor}

\definecolor{verylightgray}{rgb}{.97,.97,.97}

\lstdefinelanguage{Solidity}{
	keywords=[1]{anonymous, assembly, assert, balance, break, call, callcode, case, catch, class, constant, continue, constructor, contract, debugger, default, delegatecall, delete, do, else, emit, event, experimental, export, external, false, finally, for, function, gas, if, implements, import, in, indexed, instanceof, interface, internal, is, length, library, log0, log1, log2, log3, log4, memory, modifier, new, payable, pragma, private, protected, public, pure, push, require, return, returns, revert, selfdestruct, send, solidity, storage, struct, suicide, super, switch, then, this, throw, transfer, true, try, typeof, using, value, view, while, with, addmod, ecrecover, keccak256, mulmod, ripemd160, sha256, sha3}, % generic keywords including crypto operations
	keywordstyle=[1]\color{blue}\bfseries,
	keywords=[2]{set,address, bool, byte, bytes, bytes1, bytes2,
	bytes3, bytes4, bytes5, bytes6, bytes7, bytes8, bytes9,
	bytes10, bytes11, bytes12, bytes13, bytes14, bytes15, bytes16,
	bytes17, bytes18, bytes19, bytes20, bytes21, bytes22, bytes23,
	bytes24, bytes25, bytes26, bytes27, bytes28, bytes29, bytes30,
	bytes31, bytes32, enum, int, int8, int16, int24, int32, int40,
	int48, int56, int64, int72, int80, int88, int96, int104,
	int112, int120, int128, int136, int144, int152, int160,
	int168, int176, int184, int192, int200, int208, int216,
	int224, int232, int240, int248, int256, mapping, string, elem,
	uint, uint8, uint16, uint24, uint32, uint40, uint48, uint56,
	uint64, uint72, uint80, uint88, uint96, uint104, uint112,
	uint120, uint128, uint136, uint144, uint152, uint160, uint168,
	uint176, uint184, uint192, uint200, uint208, uint216, uint224,
	uint232, uint240, uint248, uint256, var, void, ether, finney,
	szabo, wei, days, hours, minutes, seconds, weeks, years,
	t_hash, t_batch_tag, t_public_key, element},	% types; money and time units
	keywordstyle=[2]\color{teal}\bfseries,
	keywords=[3]{block, blockhash, coinbase, difficulty, gaslimit, number, timestamp, msg, gas, sender, value, now, tx, gasprice, origin, add, epochinc, get, setminus, emptyset},	% environment variables
	keywordstyle=[3]\color{violet}\bfseries,
	identifierstyle=\color{black},
	sensitive=false,
	comment=[l]{//},
	morecomment=[s]{/*}{*/},
	commentstyle=\color{gray}\ttfamily,
	stringstyle=\color{red}\ttfamily,
	morestring=[b]',
	morestring=[b]"
}

\usepackage{listings, xcolor}

\definecolor{verylightgray}{rgb}{.97,.97,.97}
\definecolor{darkorange}{rgb}{1, 0.549, 0}
\definecolor{darkred}{rgb}{0.545, 0, 0}
\definecolor{darkblue}{rgb}{0, 0, 0.545}

\lstdefinelanguage{Frugal}{
	keywords=[1]{on, player, assert, moves_to, if, then, else,
	playsGame, Game}, %keywords 
	keywordstyle=[1]\color{blue}\bfseries,
	keywords=[2]{
	nullPreference, optimality,effect,
	init, pos, effectDivide, effectSelect, minimality, arbitrage,
	trade,
        challengedSubpath, bisectedSubpath}, %positions
	keywordstyle=[2]\color{teal}\bfseries,
	keywords=[3]{current_player,
	next_player, pl},
        keywordstyle=[6]\color{darkred}\bfseries,
	keywords=[6]{challenger},
        keywordstyle=[7]\color{darkorange}\bfseries,
	keywords=[7]{defender, proposer},
	keywordstyle=[3]\color{violet}\bfseries,
        keywords=[4]{challengeNtValidity, challengeNullPreference,
        challengeOptimality, challengeNetworkTradeVector,
        challengeEffect, challengeValidTrade, challengeL2State,
        challengeAMMState, challengeUBelowBound, challengeUtility,
        challengeNullNetworkTrade, challengeNoBetterEnoughUtility,
        challengeUtilityDefender, challengerUtilityChallenger,d1, d2,
        provideChildren, selectRight, selectLeft, challengeParent,
	provChildren, selectL, selectR,chL2state, chEffect,
	chNtValidity, chU, chU', chUBound, chBetterU, chNTVector,
	chNullPrefence, chNull, chOpt, chValidityTrade,chAMMState,
	chParent, defendEffect, chProfit, chBetterP, chArbitrage,
	defendArbitrage, chSmaller, chMin, chTrade, defNull, chP,
        bisectSubpath, selectSubpath, revealSibling},
        keywordstyle=[4]\color{darkblue}\bfseries,        
        keywords=[5]{L2BlockGame, stepGame, ntValidityGame, 
        membershipProofGame, L2StateGame, effectGame, G,
        arbitrageGame, arbStepGame, multistepMembership},
        keywordstyle=[5]\color{magenta}\bfseries,
	identifierstyle=\color{black},
	sensitive=false,
	comment=[l]{//},
	morecomment=[s]{/*}{*/},
	commentstyle=\color{gray}\ttfamily,
	stringstyle=\color{red}\ttfamily,
	morestring=[b]',
	morestring=[b]"
}

\author{Margarita Capretto}
\orcid{0000-0003-2329-3769}
\affiliation{\institution{IMDEA Software Institute}\country{Spain}}
\affiliation{\institution{Universidad Politécnica de Madrid}\country{Spain}}
\email{margarita.capretto@imdea.org}
\author{Martín Ceresa}
\orcid{0000-0003-4691-5831}
\affiliation{\institution{Input Output}\country{Spain}}
\email{martin.ceresa@imdea.org}
\author{Antonio {Fernández Anta}}
\orcid{0000-0001-6501-2377}
\affiliation{\institution{IMDEA Software Institute}\country{Spain}}
\affiliation{\institution{IMDEA Networks Institute}\country{Spain}}
\email{antonio.fernandez@imdea.org}
\author{Pedro Moreno-Sanchez}
\orcid{0000-0003-2315-7839}
\affiliation{\institution{IMDEA Software Institute}\country{Spain}}
\affiliation{\institution{VISA Research}\country{USA}}
\affiliation{\institution{MPI-SP}\country{Germany}}
\email{pedro.moreno@imdea.org}
\author{César Sánchez}
\orcid{0000-0003-3927-4773}
\affiliation{\institution{IMDEA Software Institute}\country{Spain}}
\email{cesar.sanchez@imdea.org}

\newif\ifrevision
\revisionfalse
\usepackage[normalem]{ulem}
\newcommand{\removetextrevision}[1]{\ifrevision \textcolor{red}{\sout{#1}}\else{}\fi}
\newcommand{\addtextrevision}[1]{\ifrevision \textcolor{blue}{#1}\else #1\fi}

\begin{document}

\title{ \bf A Secure Sequencer and Data Availability Committee for
  Rollups (Extended Version)}

\author{}
% \affiliation{%


\begin{abstract}
  %
  % What is the problem. Justify why the problem is a problem.
  %
  Blockchains face a scalability limitation,\removetextrevision{ caused by the
  intrinsic throughput limitations of consensus protocols and the
  desire of decentralization.}\hypertarget{consensus-limitations:abstract}{}
  \addtextrevision{partly
    due to the throughput limitations of consensus protocols,
    especially when aiming to obtain a high degree of
    decentralization.}
  Layer 2 Rollups (L2s) are a faster alternative to conventional
  blockchains.
  L2s perform most computations offchain using
  \addtextrevision{minimally} blockchains (L1)
  \removetextrevision{minimally} under-the-hood to guarantee
  correctness.
  % 
  %In L2s,
  A \emph{sequencer} is a service that receives offchain L2
  transaction requests, batches these transactions, and commits
  compressed or hashed batches to L1.
  Using hashing needs less L1 space\removetextrevision{, }\addtextrevision{---}which is beneficial for gas cost\removetextrevision{, }\addtextrevision{---}but requires a data availability committee (DAC) service \removetextrevision{in charge
  of translating}\addtextrevision{to translate} hashes into their corresponding batches of
  transaction requests.
  The behavior of sequencers and DACs influence the evolution of the
  L2 blockchain, presenting a potential security threat and delaying
  L2 adoption.

  We propose in this paper fraud-proof mechanisms, arbitrated by L1
  contracts, to detect and generate evidence of dishonest behavior of
  the sequencer and DAC.
  We study how these fraud-proofs limit the power of adversaries that
  control different number of sequencer and DACs members, and provide
  incentives for their honest behavior.
  We designed these fraud-proof mechanisms as two player games.
  \hypertarget{fp-comparison:abstract}{} \addtextrevision{Unlike the
    generic fraud-proofs in current L2s (designed to guarantee the
    correct execution of transactions), our fraud-proofs are over
    pre-determined algorithms that verify the properties that
    determine the correctness of the DAC.
    Arbitrating over concrete algorithms makes our fraud-proofs more
    efficient, easier to understand, and simpler to prove
    correct.}\marga{not referenced in table}
  We provide as an artifact a mechanization in LEAN4 of our
  \addtextrevision{fraud-proof} games, including (1) the verified strategies that honest
  players should play to win all games as well as (2) mechanisms to
  detect dishonest claims.
  \end{abstract}

\maketitle

\section{Introduction}
\label{sec:intro}
% 
% Blockchan and Scalability
%
\emph{Distributed ledgers} (also known as \emph{blockchains}) were
first proposed by Nakamoto in 2009~\cite{nakamoto06bitcoin} in the
implementation of Bitcoin, as a method to eliminate trusted third
parties in electronic payment systems.
A current major obstacle for a faster widespread adoption of
blockchain technologies in some application areas is the limited
scalability of\removetextrevision{modern blockchains}
\hypertarget{modernblockchains}{}\addtextrevision{smart contract-enabled blockchains}.
This is due to (1) the limited throughput\removetextrevision{inherent
to} \hypertarget{consensus-limitations:intro}{}\addtextrevision{of} Byzantine
consensus
algorithms~\cite{Croman2016ScalingDecentralizedBlockchain,Tyagi@BlockchainScalabilitySol},
and (2) the limitation in the block size due to the desire of
decentralization.
\addtextrevision{Issue (1)}\removetextrevision{The former} limits the number of blocks per second, while the
decentralized validation limits the size of the
blocks~\addtextrevision{\cite{ethBlockSize}}.\hypertarget{cite-ethblocksize}{}
For example, Ethereum~\cite{wood2014ethereum}---one of the most
popular blockchains---is limited to less than 4 blocks per minute,
each containing less than two thousand transactions.

%
% What is L2
%
Layer 2 (L2) rollups provide a faster alternative to blockchains, like
Ethereum, while still offering the same interface in terms of smart
contract programming and user interaction.
%
% L2 Scalability
%
L2 rollups seek to perform as much computation as possible offchain
with the minimal blockchain interaction---in terms of the number and
size of invocations---required to guarantee a correct and trusted
operation.
%
% How L2 works
%
L2 rollups work in two phases.
In the first phase, users inject transaction requests communicating
with a service called a \emph{sequencer}, which orders the transaction
requests and packs them into batches.
After creating a batch, the sequencer compresses the batch and injects
the result into the underlying blockchain (L1).
Once the batch is posted to L1 the transaction order is determined.
In the second phase, the effects of executing transaction batches are
computed offchain by agents called State Transition Functions~(STFs).
STFs are independent parties that compute L2 blocks from batches and
post the resulting state (that is, the state of the L2 blockchain)
into L1.
There are two main categories of L2 rollups:
\begin{compactitem}
\item \emph{ZK-Rollups:} STFs post zero-knowledge proofs that
  encode the correctness of the result obtained after computing the
  transactions in the batch. These proofs are verified by the L1
  contract. Upon successful validation the new L2 block consolidates.
\item \emph{Optimistic Rollups:} STFs post L2 blocks which are
  optimistically assumed to be correct, delegating block validation on
  fraud-proof mechanisms.
\end{compactitem}
%\end{compactitem}
% 
The most prominent Optimistic Rollups based on their market
share~\cite{l2beat} are Arbitrum One~\cite{ArbitrumNitro},
Base~\cite{base} and OP mainnet~\cite{optimism}.
Popular ZK-Rollups include Starknet~\cite{starknet} and zkSync
Era~\cite{zksyncera}.% and Linea~\cite{linea}.

Optimistic Rollups include an arbitration process to solve
\emph{disputes}.
STFs place a stake when they propose a new L2 block into L1.
Since L2 blocks could be incorrect (that is, \addtextrevision{they
  could maliciously encode incorrectly be outcome of executing the
  transactions}\removetextrevision{they could be erroneous or
  dishonest claims}), competing STFs are given a fixed interval of
time to challenge proposed blocks.
When an STF challenges a block it also places a stake.
The arbitration process is a game governed by an L1 smart contract,
which is played between competing STFs (the proposer and the
challenger).
\hypertarget{fp-comparision:intro-context}{}\addtextrevision{The game consists on bisecting the execution trace of
  transactions in the disputed L2 block until a dispute over a single
  instruction is reached, which can be directly verified in L1.
  In other words, the game arbitrates over the finite execution of an
  arbitrary program (the sequence of smart contract executed).}
The arbitration process ensures that a single honest participant can
always win the dispute, if it plays properly.
%
% That is, if the proposer is honest it can always defend the legality
% of the block, and otherwise a single challenger can always remove an
% illegal L2 block.
%
The losing party loses the stake and the winner receives a portion of
the loser's stake as a compensation.
Note that winning a game does not necessarily mean that the winner
player is right (as an agent can also play incorrectly or stop playing
and lose on purpose).
Therefore, L2 blocks in optimistic rollups consolidate when it has
stakes after all challenges end.
%the
%challenging interval ends and not dispute is pending.
%
This guarantees that a single honest participant can enforce that the
L2 blockchain evolution only contains honest blocks.

% Optimium and Validiums
%
To increase scalability even further, the sequencer in some modern L2
rollups posts hashes of batches---instead of compressed
batches---dramatically reducing the size of the L1 blockchain
interaction.
However, using hashes to encode batches requires an additional data
service---called \emph{data availability committee} (DAC)---to
translate hashes back into their corresponding batches.
%
% If either the sequencer or the DAC is centralized the solution is
% still not a fully decentralized L2 rollup.

ZK-Rollups that rely on DACs are known as \emph{Validiums}, which
include Sophon~\cite{sophon} and Lens~\cite{lens}.
% Immutable X~\cite{immutablex}, Astar zkEVM~\cite{astar}, and X
% Layer~\cite{xlayer}.
%
Optimistic Rollups that use DACs are called \emph{Optimiums}, such as
Arbitrum Nova~\cite{ArbitrumNitro}.~\footnote{See~\cite{l2beat} for a
  complete list of L2 rollups on top of Ethereum.}

To simplify notation, in the rest of the paper we simply use L2 to
refer to either Optimistic Rollups or ZK-Rollups that post hashes and
have a DAC, (i.e.  we use L2 to refer to Optimiums and Validiums).
Following~\cite{capretto2025decentralized}, \emph{we use the term
  \textbf{arranger} to refer to the combined service formed by the sequencer
  and the data availability committee}.
Therefore, the arranger service receives transaction requests, creates
batches containing many requests, commits the hash of these batches to
L1, and is responsible for translating the hashes back into batches
upon request.
\hypertarget{clarification}{}\addtextrevision{Once the arranger has posted batches and translates
  their contents, the STFs can proceed to compute the effects of
  executing the transactions in the batch. That is, arrangers}\removetextrevision{Arrangers} are common to Optimiums and Validiums because they not
execute transactions, which is where these L2s differ.

Arrangers have the power to influence both liveness and safety of the
L2.
For example, arrangers can ignore transactions or users, fail to
submit hashes or provide incorrect data that does not correspond to
batches of valid transaction requests and collude with STFs to post
invalid L2 blocks that are indisputable.
Arrangers can also try to delay the L2 by not responding requested
translation.
In order to prevent censorship of transaction requests, some L2s
provide mechanisms to bypass the sequencer and add transaction
requests directly in L1.
However, to the best of our knowledge, there are no mechanism to
detect and prevent in L1 safety violations from arrangers.

\vspace{0.5em}
\noindent\textbf{The Problem.}
%\subparagraph{The Problem:}
%
\emph{In this paper we attack the problem of reducing the power of
  arrangers over the evolution of L2s when trust assumptions are
  violated.}

\vspace{0.5em}
\noindent\textbf{Our Solution (overview).}
%\subparagraph{The Solution.}
%
\emph{We provide a collection of fraud-proof mechanisms to detect
  violations of safety properties and to enforce the correct evolution
  of the L2, and incentives to promote correct behavior.}
\vspace{0.01em}

We adopt an open permissioned model~\cite{Crain2021RedBelly} where
permissionless L2 users can issue transaction requests to the
permissioned arranger servers (which we call \emph{replicas}).
This model can also be adapted to a permissionless setting with
committee sortition~\cite{gilad2017algorand} without significant
modifications.
We consider a refinement of the Byzantine failure
model~\cite{Lamport1982Byzantine}, in which there are two types of
faulty replicas: \emph{Byzantine} replicas and \emph{corrupt}
replicas.
Non faulty replicas are called \emph{honest}.
The difference between Byzantine and corrupt replicas is that
Byzantine replicas can behave arbitrarily at all points in time, while
corrupt replicas cannot interfere with honest replicas in the
agreement of batches.
Corrupt replicas can misbehave in other ways during the protocol, for
example they can sign invalid batches or refuse to translate hashes.
\hypertarget{corrupt:intro}{}\addtextrevision{Essentially, corrupt replicas are less powerful than Byzantine replicas.}
This distinction between Byzantine and corrupt replicas allows us to
study adversarial models with varying degrees of power and evaluate
their impact on the properties that our system guarantees.
In particular, we analyze three different adversarial models.
The simplest \removetextrevision{one of these adversaries,}\addtextrevision{adversary} does not violate the trust
assumptions of the arranger, and thus the arranger remains correct and
all safety and liveness properties hold.
In this case, only a fraud-proof mechanism to discard batches without
enough signatures is used.
At the other extreme, we consider an adversary which can control all
arranger replicas.
In theory, this compromised arranger does not offer any guarantees.
However, with our fraud-proofs any honest agent can guarantee that
only correct batches consolidate, and prove in L1 violations to safety
properties, exposing that the arranger is compromised.
Still, if the compromised arranger satisfies the safety properties but
violates the liveness property it cannot be detected or proved in L1.
Finally, we examine an adversary whose power sits in the middle of the
previous two.
This adversary \removetextrevision{control}\addtextrevision{controls} enough replicas to generate batches with the
required signatures, but at the same time honest replicas can also
agree on batches and have these batches signed with enough signatures,
\addtextrevision{because not all adversary controlled replicas are
  Byzantine, some are just corrupt}.\marga{not referenced in table}
In this case, even if all batches posted by the malicious replicas
controlled by the adversary are ``correct'' but different from batches
agreed by the honest replicas, honest replicas can expose that the
arranger is compromised.

First, we propose fraud-proof mechanisms, based on Refereed Delegation
of Computation~\cite{canetti2011practical,canetti2013refereed}, to
prove that a batch posted by the arranger is incorrect or unavailable.
If a batch or hash is proven incorrect, all replicas involved lose
their stake and the batch is discarded
\hypertarget{discarded}{}\addtextrevision{(along with all L2 blocks
  executing its transactions)}.
Additionally, our fraud-proof mechanisms can generate undeniable
proofs of fraud, which could be used as evidence for requiring to
replace faulty replicas.
More importantly, a single honest agent (with enough resources to pay
the L1 fee to play the moves of the game) can use our fraud-proof
mechanisms to enforce that incorrect or unavailable batches do not
consolidate, guaranteeing that all L2 blocks can be computed and
disputed, \emph{even if all arranger replicas are faulty.}
That is, our fraud-proof mechanisms can be used (1) to guarantee
safety properties of arrangers even when the trust assumption are
violated and also (2) as a deterrent for replicas from being faulty.
As consequence, we obtain a L2 solution where a single honest agent
can guarantee safety of the entire system, rather than just the
execution part (where the STFs compute the effects), as is the case in
current L2s.

We modeled our solution in LEAN4, and prove that a
single honest agent is enough to prevent faulty assertions.
LEAN4 is a proof-assistant aimed to close the gap between automated and
interacting theorem proving~\cite{Moura.2021.Lean4}.
The mechanization involves modeling assertions, the possible actions
agents can take, and mainly, the construction of fraud-proofs when
faulty assertions are detected.
\removetextrevision{Harnessing the computing power of LEAN4}
We can directly run our
verified strategies that honest agents can use to win all challenges
against faulty agents.

\hypertarget{fp-comparision:intro}{}
\addtextrevision{ It is important to note that, unlike fraud-proof
  mechanisms in current L2s, we propose fraud-proofs over
  pre-determined algorithms which verify specific properties.
  Fraud-proofs in current L2s are over the execution trace of arbitrary
  algorithms (as transactions can invoke arbitrary contracts during
  their execution) and require: (1) reasoning about the execution of an
  interpreter of smart contracts and (2) the ability to extract the
  state after each instruction, not just the final result, to provide
  these states in the trusted L1 arbitrator contract as required by the game.
  In contrast, our fraud-proofs exploit that we arbitrate over
  pre-determined algorithms, dividing the execution of these algorithm
  into well defined high-level blocks.
  This approach has the following advantages when compared to
  fraud-proofs for arbitrary algorithms:
  \begin{compactitem}
  \item Modularity and clarity: as each building block is
    well-defined, each move and resulting position in the game is
    easier to understand and formalize.
  \item Efficiency: our fraud-proofs can process blocks of
    instructions instead of having to reason at the instruction level,
    which reduces the total number of moves of the game.
  \item Formality: the simplicity of our fraud-proofs, compared to
    those currently used in L2s, enabled us to formalize our games and
    prove their correctness in a LEAN4 library of around 5000 lines of
    code.
    As far as we know, fraud-proofs used in Optimistic Rollups and
    Optimium have not been formalized yet and would likely be a much
    more complex endeavor.
\end{compactitem}
Although fraud-proofs are meant to be a deterrent, only to be executed
in the worst case scenario, an incorrect implementation may lead to
the existence of fraud-proofs that invalidate correct blocks or the
impossibility to generate fraud-proofs of incorrect blocks, rendering
the whole L2 scheme incorrect. To reason about fraud proofs clear and
modular approaches are needed.}

Second, we introduce \emph{incentives} for arranger replicas to
participate in the protocol, e.g. including payments for generating
and signing hashes, posting correct batches into L1, and performing
reverse translations.
These replicas place stakes when posting batch hashes and are rewarded
for batches that consolidate.

By combining incentives with fraud-proof mechanisms, we create a
motivation for rational agents to behave honestly.
To the best of our knowledge, this is the first work to introduce
incentives and fraud-proof mechanisms for arrangers (sequencer and
DACs) of L2s.
Our incentives and fraud-proof mechanisms are general and do not
depend on concrete implementation of the arranger \addtextrevision{or
  how the execution part is handled by the STFs}.
\hypertarget{implementation}{}
\addtextrevision{Therefore, current L2s can easily integrate our
  protocol.
  L2s just need to (1) replace the existing contract that receives
  batches from their arranger and (2) deploy new contracts to govern
  the fraud-proofs of our protocol.
  These new fraud-proofs focus on properties of the batches posted by
  the arranger, rather than transaction execution, ensuring that they
  do not conflict with existing fraud-proofs.
}
\imarga{Add Arbitrum as example?}

 % For example, in Arbitrum it would be required to replace the
 % “Sequencer Inbox” contract to keep track of posted batches,  and
 % their status (consolidated, challenged or pending). The remaining
 % contracts do not need to change.
 
\paragraph{\textbf{Contributions.}}
In summary, the contributions of this paper are:
\begin{enumerate}[(1)]
  \hypertarget{contribution}{}
\item \addtextrevision{Novel fraud-proofs over pre-determined
    algorithms which verify specific arranger properties in
    Sections~\ref{sec:fraudproofs}};
\item Economic incentives including payments and fraud-proof
  mechanisms to detect protocol violations and punishments for
  faulty replicas in Sections~\ref{sec:fraudproofs} and~\ref{sec:incentives};
\item An analysis of three adversary models and their limitations and
  impacts on the evolution of the L2 blockchain in Section~\ref{sec:threat-models};
\item One artifact: a library in LEAN4 mechanizing fraud-proof
  mechanisms proving that honest players always win.
\end{enumerate}

\paragraph{\textbf{Structure.}}
The rest of the paper is organized as follows.
Section~\ref{sec:prelim} states our assumptions about L1s, presents
Optimistic Rollups and Optimiums, describes the computation model, and briefly presents Merkle trees.
Section~\ref{sec:api} describes the concept of arranger.
Section~\ref{sec:solution-overview} gives an overview of our proposal
to limit the power of arrangers using: (1) fraud-proof mechanisms,
explained in detail in Section~\ref{sec:fraudproofs}, and (2) economic
incentives, studied in Section~\ref{sec:incentives}.
Section~\ref{sec:threat-models} presents three threat models and
analyzes their impact in the evolution of L2s.
Section~\ref{sec:related-work} compares with related work.
Finally, Section~\ref{sec:conclusion} concludes.

%%% Local Variables:
%%% TeX-master: "main.tex"
%%% TeX-PDF-mode: t
%%% End:

\section{Definitions. Model of Computation}\label{sec:prelim}

We state now our assumptions about L1s, briefly present Optimistic
Rollups and Optimiums, describe the computation model and present an
overview of Merkle trees.

\subsection{Assumptions about L1}

We assume that the L1 ensures both liveness and safety.
Specifically, while the system tolerates temporary censorship of L1
transaction requests and reordering of L1 transactions, it guarantees
that every transaction submitted to L1 is eventually processed
correctly.
The L1 includes the following specific smart contracts:
\begin{itemize}
\item a \<logger> that arranger replicas use to post batches of
  transaction requests, and
\item a set of smart contracts that arbitrate fraud-proof mechanisms
  (see Section~\ref{sec:fraudproofs}).
\end{itemize}

\subsection{L2 Optimistic Rollups and Optimiums}

\begin{figure}[t]
  \centering
  \begin{tabular}{c}
    \includegraphics[scale=0.32]{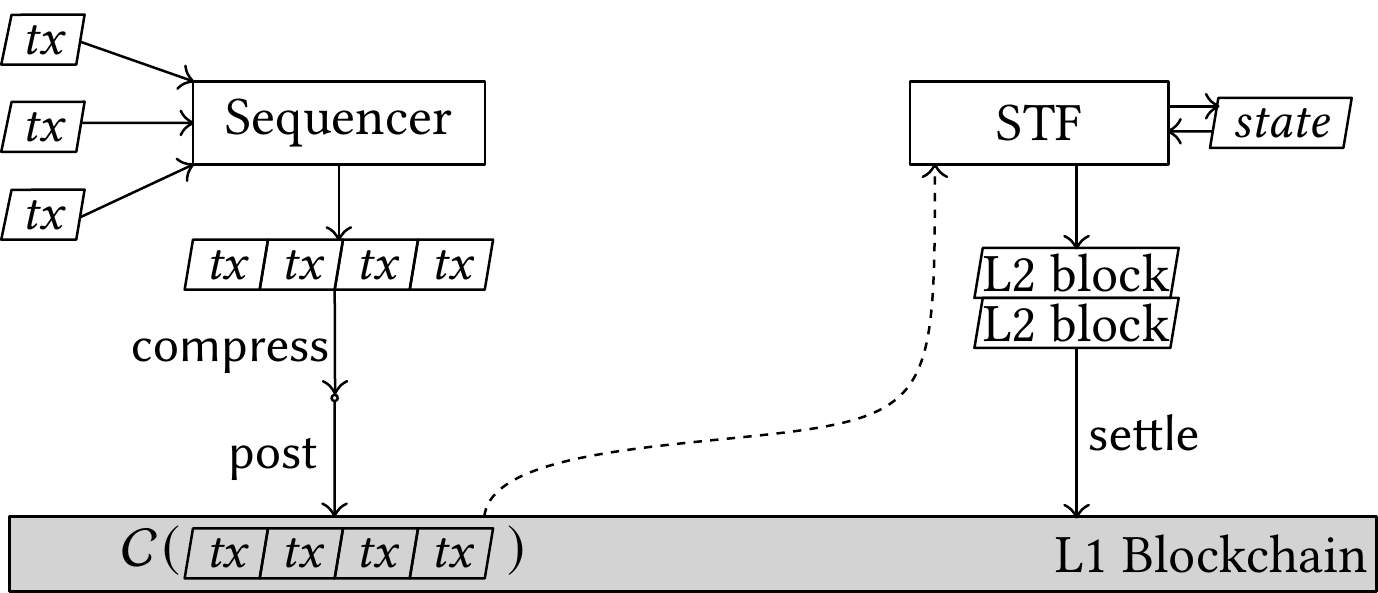}
  \end{tabular}
  \caption{Optimistic Rollups.}
    \label{fig:arbitrum-nitro}
  \end{figure}
\begin{figure}[t]
  \centering
  \begin{tabular}{c}
     \includegraphics[scale=0.32]{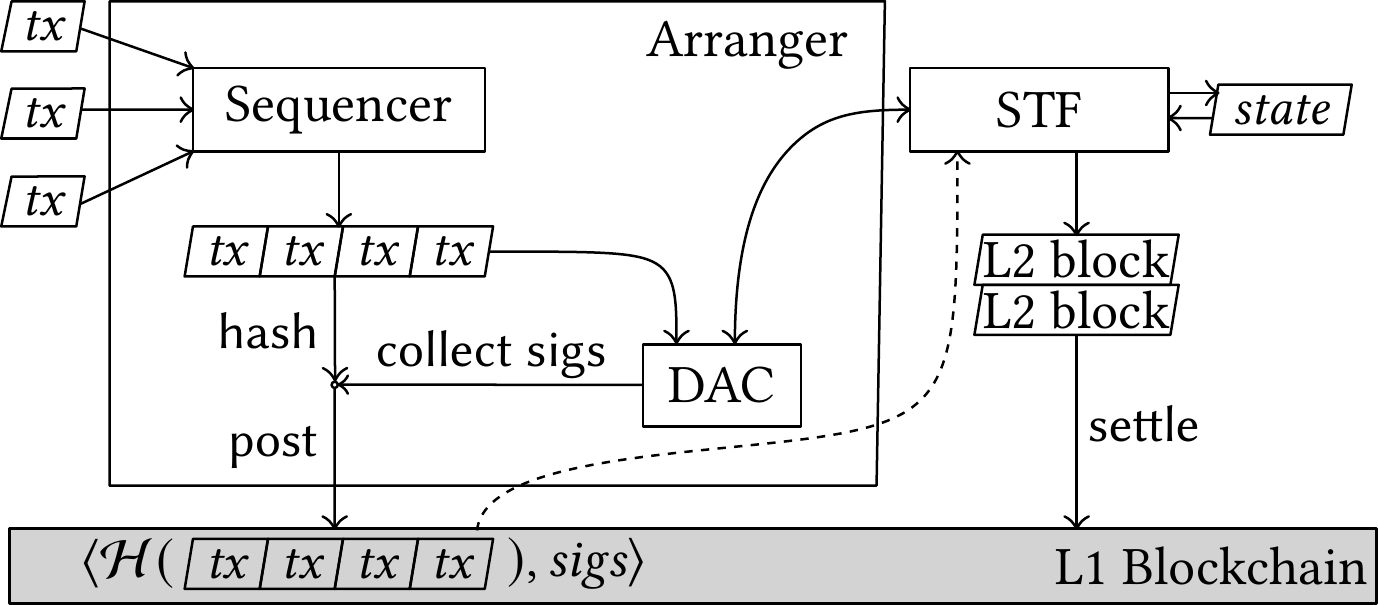}
  \end{tabular}
  % \caption{Sequencer with a DAC, like Arbitrum AnyTrust and this paper.}
  \caption{Optimiums.}
  \label{fig:arbitrum-anytrust}
\end{figure}

L2 Optimistic Rollups split transaction sequencing
from transaction execution.
The consolidation of transaction effects is delayed to allow disputes
and arbitration.
L2 Optimistic Rollups are implemented as two components (see
Fig.~\ref{fig:arbitrum-nitro}):
a \emph{sequencer}, in charge of ordering transactions, and a
\emph{state transition function} (STF), responsible for executing
transactions (terms introduced by
Arbitrum~\cite{Kalodner2018Arbitrum}).
While currently the role of the sequencer is centralized, anyone can
in principle be an STF.\footnote{In Arbitrum One STFs are allow-listed~\cite{ArbitrumDecentralization}.}

The sequencer collects transaction requests from L2 users, packs them
into batches, and posts these batches as a \emph{single} invocation
into L1.
Once posted, compressed batches are immutable and visible to everyone.
Currently, sequencers are centralized processes and do not offer any
guarantee regarding transaction orders, whether transaction requests
can be discarded or that the data posted in L1 actually correspond to
a batch of transaction requests.
To prevent censorship from the sequencer, some L2s allow users to
perform a more expensive posting of transaction requests straight to
L1, which is visible and must be executed.

% Stf + Rollup
%
STFs are independent processes in charge of computing the effects of
batches of transaction requests, which is determined by the sequence
of transaction requests and the previous L2 state.
STFs propose new L2 states as L2 block \emph{assertions} into L1.
L2 blocks are optimistically assumed to be correct, but there is a
challenge period of typically a week.
STFs place stakes in L2 blocks, asserting that the L2 block is
correct, and can also place stakes to challenge L2 block assertions
posted by other STFs.
When STFs challenge assertions, a fraud-proof mechanism is played in
L1 between the involved STFs.
The mechanism guarantees that honest players win.
The losing party forfeits their stake, while the winner receives a
portion of the loser's stake as compensation.
L2 block assertions without stakes are removed, while surviving ones
consolidate and the L2 evolves.
\hypertarget{fp:prelims}{}\addtextrevision{These fraud-proof
  mechanisms are games that bisect the execution trace of the
  transaction in the disputed L2 block, until the challenge is reduced
  to a single instruction.
  At this point, players agree on the state before the instruction but
  disagree on the state after the instruction.
  All states in the middle of challenged subtraces and the final
  instruction must be posted in the L1 contract that arbitrates the
  fraud-proof.
  Therefore, these fraud-proof mechanisms require: (1) reasoning about
  the execution of an interpreter of smart contracts and (2) the
  ability to extract the state after each instruction, not just the
  final result, to provide them in the trusted L1 contract as
  necessary.}

Sequencers can post batches in L1 after compressing them, using a
reversible compression algorithm~\cite{Alakuijala18brotli}.
However, the L1 gas associated with posting compressed batches is
still significant.
To mitigate this cost, Optimiums post hashes instead, and add a
\emph{Data Availability Committee} (DAC) storing batches and providing
them upon request (see Fig.~\ref{fig:arbitrum-anytrust}).
In current Optimiums, the sequencer posts a hash of a batch---which is
much smaller than the compressed batch---along with evidence, i.e. signatures,
% (in the form of signatures of the hash)
that the batch is available from at least one honest DAC member.
STFs then check that posted hashes have enough valid signatures and
request the corresponding batch to DAC members.
To ensure progress, some implementations provide a fallback
mechanism, where the sequencer posts the compressed batch into L1 if it
cannot collect enough signatures within a specified time frame.
In this work, following~\cite{capretto2025decentralized}, we use
\emph{arranger} to refer to the service that combines a sequencer and
a DAC.

\subsection{Model of Computation}

Our system comprises arranger replicas and arranger clients.
There are two types of clients:
\begin{itemize}
\item L2 users sending L2 transaction requests to arranger replicas,
  and
\item STFs requesting arranger replicas the translation of hashes
  (posted through the L1 \<logger> smart contract) into batches.
\end{itemize}

We consider a public-key infrastructure (PKI) that associates replica
and client identities with their public keys, and that is common to
all replicas and clients.
L2 users can create \emph{valid} transaction requests and arranger
replicas cannot impersonate clients.
Valid transaction requests are those that have been correctly signed,
so arranger replicas can locally check their validity using public-key
cryptography.

Arranger replicas can be classified as either \emph{honest}, meaning
they adhere to the arranger protocol, or \emph{faulty}.
Faulty replicas include \emph{Byzantine} replicas, which behave
arbitrarily~\cite{Lamport1982Byzantine}.
We analyze other kinds of faulty replicas in
Section~\ref{sec:threat-models}.

Arranger replicas use a known collision-resistant hash function
\(\hash{}\) to hash transaction requests and create Merkle trees from
batch of transactions requests.

\subsection{Merkle Trees}
Merkle trees~\cite{Merkle88} are a tree data structure, where each
leaf node is labeled with the hash of its content and each node that
is not a leaf is labeled with the hash of the concatenation of its
children hashes.

Merkle trees provide an efficient (logarithmic) membership
authenticated check, which we employ in our fraud-proof mechanisms
(see Section~\ref{sec:fraudproofs}).
Given the hash of the root \(r\) of a Merkle tree, also known as
\emph{Merkle root}, the proof that the content of the \(i\)-th leaf is
\(x\) consists of index \(i\), element \(x\), and the hashes of all
nodes that are neighbors of nodes in the path from the \(i\)-th leaf
to the root.
Verifying that the proof is correct involves reconstructing the hash
in all nodes in the path from the leaf to the root, which can be done
bottom-up with the data provided in the proof.
Assuming that the hash function used is collision resistant, the last
hash will match the Merkle root \(r\) if and only if \(x\) is in fact
the content of the \(i\)-th leaf.

\hypertarget{merkleroot}{} \addtextrevision{ Given a batch \(b\) we
  consider the Merkle tree that has as leaves the elements in \(b\)
  and uses \(\hash\) as hash function, and denote with \(\Mroot(b)\)
  its Merkle root.}

%%% Local Variables:
%%% TeX-master: "main.tex"
%%% TeX-PDF-mode: t
%%% End:

\section{Arranger}
\label{sec:api}

In this section we briefly explain the concept of arranger, introduced
in~\cite{capretto2025decentralized}, which seamlessly fits into the
model of existing Optimiums (see Fig.~\ref{fig:arbitrum-anytrust}).
Arrangers perform two main functions:
\begin{compactitem}
\item Sequencing: ordering and batching transaction requests, and then
  posting the corresponding hashes to L1.
\item Data Availability: ensuring the availability of data
  corresponding to the posted hashes, enabling the reconstruction of
  batches when needed. Arranger replicas offer an end point
  $\<translate>(\<id>,\<h>)$ to translate hashes.
\end{compactitem}

The arranger service receives transaction requests from L2 users.
When the arranger collects sufficient transaction requests (or a
timeout occurs) all honest arranger replicas agree on a new batch $b$
and assign an identifier $\<id>$ to \(b\).
Then, all honest arranger replicas compute hash $h$, the root of the
Merkle Tree~\cite{Merkle88} whose leaves are the transaction requests
in \(b\), and create a \emph{batch tag} $(\<id>,h)$.
Once $b$ has been agreed, $(\<id>,h)$ is computed locally.
Each honest replica then signs the new batch tag and a compressed
version of the batch and propagates its signatures to all other
replicas.
Once enough signatures of batch tag are collected, a combined
signature \(\sigma\) is generated (in our system, we use
BLS~\cite{Boneh2001Short}).
The resulting \emph{signed batch tag} $(\<id>,\<h>,\sigma)$ is then
posted to L1.
The second signature, of the compressed version of the batch, is used
in a fraud-proof game that guarantee that the batch is available (see
Section~\ref{sec:fraudproofs:data-availability}).
\hypertarget{identifiers}{}\addtextrevision{An attempt to post a
  batch to L1 becomes an L1 transaction request, which can be
  reordered.
  Therefore, we use identifiers in signed batch to order batches.
  In particular, two consecutive batches agreed by honest arranger
  replicas have consecutive identifiers.}

The \<logger> L1 smart contract accepts signed batch tags without
performing any validation check.
In this paper, we propose fraud-proof mechanisms---similar to
arbitration protocols in Optimistic Rollups---to discard incorrect or
unavailable signed batch tags (see Section~\ref{sec:fraudproofs}).

STFs monitor the \<logger> contract and, after locally validating the
signatures of signed batch tags, request the corresponding batch from
the arranger to compute the next L2 block.
If the signature is invalid, the batch cannot be retrieved or the
batch contains invalid transaction requests, the STFs can use a
fraud-proof mechanisms to discard the batch tag and penalize
misbehaving arranger replicas.

\subsection{Arranger Properties}
\label{sec:api-properties}
Before given the properties of \emph{correct} arrangers, we first
introduce some definitions about signed batch tags.
Signed batch tags are considered \emph{certified} when they have at
least \(\s{}\) signatures of arranger replicas, where \(\s\) is a
static system parameter known by the arranger replicas, the \<logger>
contract and the set of L1 smart contracts arbitrating fraud-proof
mechanisms.
A certified batch tag $(\<id>,\<h>,\sigma)$ is \emph{legal} if its
corresponding batch \(b\) satisfies the following properties:
% with batch $b$ is legal whenever
% it satisfies the following properties:
\begin{compactitem}
\item \PrValidity: Every transaction request in $b$ is a valid
  transaction request added by \removetextrevision{a client}\hypertarget{validity}{}\addtextrevision{an L2 user}.\footnote{A transaction request is valid
    when it is properly formed and signed by the originating \removetextrevision{client}\addtextrevision{L2 user}.}
  \item \PrIntegrityOne: No transaction request appears twice in $b$.
  \item \PrIntegrityTwo: No transaction request in $b$ appears
    in\hypertarget{integrityTwo}{}\addtextrevision{ a batch corresponding with} a legal batch tag
    previously posted by the arranger.
  \end{compactitem}

All certified batch tags posted by correct arrangers must be legal.

\begin{property}[\textup{\PrLegality}]
  \label{pr:legality}
  Every certified batch tag posted by the arranger is a legal batch
  tag.
\end{property}

Arranger replicas can post multiple signed batch tags with the same
identifier and the \<logger> accepts them all.
A batch tag can be part of two signed batch tags if each tag is signed
by a different subset of arranger replicas.
However, to ensure deterministic evolution of L2, two certified batch
tags with the same identifier must correspond to the same batch, and
thus have the same hash.
Formally:

\begin{property}[\textup{\PrUniqueBatch}]
  \label{pr:uniqueBatch}
  Let $(\<id>, \<h>_1, \sigma_1)$ and $(\<id>, \<h>_2, \sigma_2)$ be
  two posted certified batch tags.
  Then, \(\<h>_1 = \<h>_2\).
\end{property}

To ensure censorship resistance and data availability, correct
arrangers must also satisfy the following properties:

\begin{property}[\textup{\PrTermination}]
  \label{pr:termination}
  All valid transaction requests added to honest replicas eventually
  appear in a posted legal batch tag.
\end{property}

\begin{property}[\textup{\PrAvailability}]
  \label{pr:availability}
   Every posted legal batch tag can be translated into
   its batch.
 \end{property}
 
\PrAvailability is expressed formally as follows. Let
 $(\<id>, \<h>, \sigma)$ be a legal batch tag posted by the arranger,
 s.t. \removetextrevision{$\<h>=\hash(b)$}\hypertarget{mroot:use1}{}\addtextrevision{$\<h>=\Mroot(b)$}.
 Then, some honest replica will return $b$ when requested
 $\<translate>(\<id>,\<h>)$.
 This prevents halting the L2 blockchain and having indisputable L2
 blocks by failing to provide batches of transaction requests from
 hashes.
 If a faulty replica returns a batch $b'$ with \(b' \neq b\), by the
 assumption of collision resistance, the root of the Merkle tree
 corresponding with \(b'\) cannot be \(\<h>\).
 Clients can locally verify this mismatch by computing
 \removetextrevision{$\<h>=\hash(b')$}\hypertarget{mroot:use2}{}\addtextrevision{\(\<h>=\Mroot(b')\)}.

%\PrValidity, \PrIntegrityOne,
\PrLegality, \PrUniqueBatch and \PrAvailability are safety properties
and \PrTermination is a liveness property.
Altogether, they characterize correct arrangers.
%
% Transactions are created by clients and posted once.
% %
% Since arranger processes employ hashes to optimize storage and
% communication costs, we need to guarantee there is always a way to
% retrieve data from \newtext{legal batch tags}.
%correct hashes
%
%
\hypertarget{correctArrangers}{}Correct arrangers offer censorship resistance because all valid
transaction requests added to honest arranger replicas are eventually
executed, guarantee that no transaction is executed twice and that
only valid transaction requests are executed.
For more details, see~\cite{capretto2025decentralized}.

% \subsection{Arranger Implementations and Their Trust Assumptions}\label{sec:arranger:implementations}
\subsection{Arranger Implementations}\label{sec:arranger:implementations}

A straightforward implementation of an arranger, which we refer to as
\emph{centralized}, consists of a single replica that performs all
actions of arrangers: receiving transaction requests from L2 users,
packing them into batches, \hypertarget{identifiers:centralized}{}\addtextrevision{assigning unique identifiers}, posting
hashed batches in L1 and translating hashes back into batches of
transaction requests.
Clearly, the correctness of a centralized arranger depends entirely on
its single replica.

Most existing L2s do not implement a centralized arranger but instead
what we call a \emph{semi-decentralized} arranger.\addtextrevision{\footnote{For the
  current decentralization status of L2s running on top of Ethereum
  see Appendix~\ref{app:decentralizationL2Eth}}}
In this architecture, a single replica acts as the sequencer and the
remaining replicas implement a decentralized DAC.
The centralized sequencer collects and batches transaction requests,
\hypertarget{identifiers:semidecentralized}{} \addtextrevision{assigns
  unique identifiers to create batch tags}, communicates these batches
\addtextrevision{and batch tags} to all DAC replicas, collects their
signatures and posts signed batch tags to the \<logger> in L1.
DAC replicas then can provide the inverse resolution of hashes posted
by the sequencer.
To ensure correctness, implementations of a semi-decentralized
arranger assume that the sequencer is honest and the number of faulty
replicas in the DAC is less than the signature threshold \(\s{}\).
The concrete value of \(\s{}\) and its relation with the number of
total replicas \(n\) vary depending on the system configuration.

To completely remove the single point of \removetextrevision{of
  failure and trust the}\addtextrevision{trust and failure that is
  the centralized} sequencer\removetextrevision{still represents} in semi-decentralized
  arrangers, \cite{capretto2025decentralized} proposes a \emph{fully
    decentralized} arranger based on Set Byzantine Consensus
  (SBC)~\cite{Crain2021RedBelly}.
SBC is a\removetextrevision{more efficient (in term of throughput) variant of
  Byzantine Consensus}
\hypertarget{SBC:throughput}{}\addtextrevision{variant of Byzantine Consensus that provides
  high throughput}
where replicas agree on a set of elements, instead of a
single element.
In this implementation all replicas perform both the roles of
sequencer and DAC member, and each honest replica includes an honest
SBC replica as a building block.
L2 users can add transaction requests through any replica.
Set consensus guarantees that all honest replicas eventually agree on
the same batch and that elements added to honest replicas eventually
appear in a set agreed by honest replicas.
Then each replica can locally compute the
hash\hypertarget{identifiers:decentralized}{}\addtextrevision{ of the
  batch agreed by consensus, use the SBC instance as the batch
  identifier, and thus create a unique batch tag}, sign the batch tag, and later
translate hashes back into batches.
\hypertarget{SBC:assumption}{}\removetextrevision{As with Byzantine consensus,} SBC assumes that less than one-third of all replicas
are Byzantine.
This assumption is also inherited by the fully decentralized arranger
in~\cite{capretto2025decentralized}, which also requires that the
number of signatures required in certified batch tags is at least
\(\s{} > n/3\), guaranteeing that all certified batch tags are signed
by an at least one honest replica and ensuring the correctness of the
arranger.

%%% Local Variables:
%%% TeX-master: "main.tex"
%%% TeX-PDF-mode: t
%%% End:

\section{Solution Overview}\label{sec:solution-overview}

Arranger implementations assume that a portion of their replicas are
honest in order to ensure correctness.
However, they neither offer a mechanism to detect when this trust
assumption is violated, nor provide guarantees in scenarios where the
trust assumption does not hold.

This can have undesired consequences on the evolution of the L2.
For example, violations to property \PrAvailability can lead to all L2
tokens to be stolen.
Consider a fully-decentralized arranger that is correct under the
assumption that at most \(f < \s{}\) replicas are faulty, and an
adversary that breaks this trust assumption by controlling \(\s\)
arranger replicas.
In such case, the replicas controlled by the adversary can generate
certified batch tags that are known only by these replicas that the
adversary controls.
They can post a batch tag in the \<logger> contract and then refuse to
translate the hash into a batch of transaction requests, violating
property \PrAvailability.
As consequence, the adversary could control an STF that posts a new L2
block encoding a state where all L2 tokens are stolen.
Since other STFs do not know the corresponding transaction requests,
they are unable to compute the correct L2 block and properly challenge
the malicious block.
To address this problem, in this paper we propose fraud-proof
mechanisms (see Section~\ref{sec:fraudproofs}) that can be used to
reject batch tags that are illegal or unavailable.
These fraud-proofs can also be used to generate evidence that an
arranger does not satisfy the safety properties, discouraging faulty
behavior by replicas.
Furthermore, we provide economic incentives (see
Section~\ref{sec:incentives}) to motivate arranger replicas to remain
active.

Specifically, arrangers replicas post batch tags in L1 as
\emph{proposals}.
Similar to L2 blocks assertions, batch tag proposals are
optimistically accepted, but there can be challenges during a
challenging period.
During this period, arranger replicas and STFs can stake on a batch
tag claiming (1) that the batch tag is legal and unique, and (2) that
they can translate the hash.
Replicas can also put stakes to challenge other batch tags proposals.
There are different kinds of challenges, depending on the claim
disputed.
For each challenge, there is a fraud-proof game, arbitrated by an L1
contract, played between the replicas involved.
Staking agents that fail to defend their claim lose their stake,
either because their claim was false or because they did not
participate correctly in the fraud-proof game.
A batch tag is discarded (along with all L2 blocks executing its
transactions) when it has no stake.
Conversely, a batch tag consolidates when at least one of its staker
survives all challenges.

The different fraud-proof mechanisms (FP) for batch tags are:
% \begin{enumerate}
\begin{compactenum}
\item Certifiability \<fpm>: disputes the legality of the batch tag,
  claiming that the signature in the batch tag does not contain at least
  \(\s{}\) valid arranger replica signatures.
\item Validity \<fpm>: disputes the legality of the batch tag,
  claiming that it contains an element that is not a valid transaction
  request.
\item Integrity \<fpm> 1: disputes the legality of the batch tag,
  claiming that it contains a duplicate element.
\item Integrity \<fpm> 2: disputes the legality of the batch tag,
  claiming that it contains an element that appears in a previous
  consolidated batch tag.
\item Uniqueness \<fpm>: disputes the uniqueness of the batch tag,
  claiming that there is another certified batch tag with the same
  identifier but different hash.
 \item Data availability \<fpm>: forces the batch to be revealed
  in L1.
  % \end{enumerate}
\end{compactenum}
All these challenges, as well as the strategies for honest players,
are explained in detail in Section~\ref{sec:fraudproofs}.
The first five challenges guarantee that an staker making the correct
claim can win the challenge.
The data availability challenge ensures that either the
transaction requests corresponding to the challenged batch tag are
revealed or the batch tag is discarded (see
Proposition~\ref{prop:data-availability}).
From this, we derive the following key lemma:

\begin{lemma}\label{l:consolidate}
  A single honest agent can ensure that a batch tag consolidates if
  and only if it is legal, available and, at the end of its challenge
  period, there is no other certified batch tag posted in L1 with the
  same identifier but different hash.
\end{lemma}

Observe that if an arranger violates one of its safety properties, it
means that the arranger either (1) posted a batch tag that is either
not legal or unavailable, or (2) posted two certified batch tags with
the same identifier but different hash.
By the previous lemma, a single honest agent can use fraud-proof
mechanisms to discard such batch tag.\hypertarget{discardedbatch}{}\addtextrevision{\footnote{In the
    second case, if neither batch tag has consolidated yet, then both
    are discarded. Otherwise, the last one posted is discarded, as it
    is the one being challenged.}  }
Moreover, by winning the challenge the honest agent proves that the
arranger is violating a safety property and, as consequence, force the
replacement of the arranger.

\begin{theorem}\label{cor:safety}
  A single honest agent can prove in L1 that an arranger violates a
  safety property.
\end{theorem}

Unfortunately, detecting that the arranger is violating property
\PrTermination is impossible as this is a liveness property that may
be satisfied later in the future.
To mitigate the risk of an arranger censoring transactions requests, L2
users can post their transaction requests directly in L1, although at
a higher cost.

Finally, the data availability challenge requires posting the batch
compressed, which can be expensive.
As an alternative, we provide in
Section~\ref{sec:incentives:translate} a cheaper protocol to translate
hashes into batches using zero-knowledge contingent payments.
Although this cheaper protocol is preferable for both arranger
replicas and STFs, it does not provide any guarantee when arranger
replicas remain silent and thus the data availability challenge is a
necessary fallback.

%%% Local Variables:
%%% TeX-master: "main.tex"
%%% TeX-PDF-mode: t
%%% End:

\section{Fraud-proofs Mechanisms}
\label{sec:fraudproofs}

We present now fraud-proof mechanisms (also called fraud-proof games,
or for conciseness \<fps>), governed by L1 smart contracts, which
enable a single honest agent (replicas and STFs)---with enough tokens
to participate---to prevent the consolidation of illegal or
unavailable batches, and to expose faulty behavior.
Our \<fps> do not depend on the concrete arranger implementation, the
number of faulty replicas or the correctness of the arranger.
  In Section~\ref{sec:threat-models}, we study how these mechanisms can limit the
  power of different adversaries.
We assume participating agents have public L1 accounts.
% where they can
% receive and send funds.

Our \<fpms> are similar to the ones used in optimistic rollups to
dispute L2 blocks, which are based on Refereed Delegation of
Computation (RDoC)~\cite{canetti2011practical, canetti2013refereed}.
The main difference is that our \<fps> are over concrete algorithms
that check the availability and legality of posted batch
tags, and not over the execution of arbitrary smart contracts
translated to WASM~\cite{wasm}.
Arbitrating over concrete algorithms\removetextrevision{results in
  more efficient and simpler to understand \<fpms>.}
\hypertarget{fp}{} \addtextrevision{allows us to divide the algorithms
  in well-defined high level blocks and arbitrate over these blocks
  instead of sequences of instructions that are not known upfront.
  This results into modular and clear \<fpms> that are amenable for
  formalization, which are also more efficient.
  For example, consider an algorithm that verifies a Merkle proof of a
  Merkle tree with 4096 (= $2^{12}$) transaction requests.
  The execution trace of a WASM program obtained by compiling a Rust
  program contains around $2^{21}$ instruction.
  Therefore, a bisection game over its execution trace involves 21
  moves per player, where each move requires accessing to
  intermediate states of the execution trace.
  On the other hand, the multi-step membership \<fp> that we propose
  (see Section~\ref{sec:fraudproofs:membership}) requires 4 moves per
  player (that only require knowing nodes in the Merkle tree), because
  this \<fpm> is logarithmic in the height of the Merkle tree.
  Furthermore, a one-step membership \<fp> (see
  Section~\ref{sec:fraudproofs:membership}) requires only 1 step,
  provided that the gas limit can check the Merkle proof in a single
  execution.
}

Our \<fps> are two-player games arbitrated by L1 smart contracts.
Each player takes turns with a predetermined \addtextrevision{total}
time per game, which only decreases in the player's turn (similar to
chess clocks, \hypertarget{clock}{}\addtextrevision{which is the time
  mechanism currently implemented in the FPs employed in L2s)\footnote{
    The general consensus among the Ethereum research community and L2
    projects about the duration of a game is two weeks (one week per
    player).}}.
Players must put a stake to participate.
The losing party loses its stake, while the winner can retrieve its
stake and receives a reward.\footnote{The pseudo
  code implementing the winning strategy for the honest players and
  their formalization in LEAN4~\cite{Moura.2021.Lean4} are available
  as an artifact.}

We describe our \<fpms> and provide an overview of the strategy for
honest players.
\addtextrevision{The pseudocode for contracts arbitrating the \<fps>
  and the strategy for honest players can be found
  Appendix~\ref{app:pseudocodes}. Figures illustrating the states of
  each \<fp>, the allowed moves in each state, and the data required
  for each move are in Appendix~\ref{app:figs}.} Informal proofs can
be found in Appendix~\ref{app:proofs}, and all
\addtextrevision{formal} proofs of the FP of batch tag legality are
formalized and proven correct in LEAN4.
\imarga{TODO: pseudocode strategies and comments in pseudocode of fp}

\subsection{Data Availability Fraud-Proof
  Mechanism}\label{sec:fraudproofs:data-availability}

Section~\ref{sec:incentives:translate} includes a fast and cheap
protocol, based on contingent payments, to obtain the batch
corresponding to the the hash of a batch tag.
However, arranger replicas can be slow or refuse to participate in
this protocol without any punishment.
The \emph{data availability \<fpm>} provides an alternative to force
agents with a stake in a batch tag \(t\) to publish the associated
batch of transaction requests \(b\).

Once an agent \Att initiates a data availability \<fpm> against batch
tag \(t\), any agent staking in \(t\) can respond by posting to a L1
contract (called \emph{data availability contract}) a compressed
version of \(b\) along with \(\s{}\) signatures from arranger
replicas.
The data availability contract rejects the data if its signature is
not valid or does not contain the same signatures as the batch tag,
but does not check that the posted data is the compressed version of
\(b\).
This last step is delegated to another \<fp>, called
\emph{decompress-and-hash \<fpm>}, that consists in bisecting the
execution trace of a program \Ptt that decompresses the data posted,
obtaining a list of elements, creates a Merkle tree with the list of
elements as leaves, and checks that the Merkle root is the hash in the
previously posted batch tag \(t\).
Agent \Att can initiate the decompress-and-hash \<fp> claiming that
program \Ptt does not finish successfully with the data provided by
the staking agents.
If \Att wins the decompress-and-hash challenge, all staking agents
lose their stake and the batch is discarded.
Otherwise, \Att loses its stake.
Decompress-and-hash is similar to the arbitration game using in
optimistic rollups for L2 blocks, but instead of arbitrating over the
execution of an arbitrary smart contract, the decompress-and-hash
\<fp> uses a fixed algorithm, whose execution trace is
shorter.\footnote{The execution trace of a WASM program obtained by
  compiling a Rust program that decompress a string representing a
  compressed batch of contains $4000$ transaction requests requires
  less than $2^{27}$ instructions, significantly less than usually
  expected in L2 blocks in practice.}
Since this \<fpm> extends RDoC~\cite{canetti2011practical,
  canetti2013refereed}, we can conclude the following result.

\begin{corollary}\label{corollary:rdoc}
  Agent \Att has a winning strategy for the decompress-and-hash \<fpm>
  if and only if the data provided by staking agents does not
  correspond to a compressed batch.
\end{corollary}

\subsubsection*{Honest Strategies}

After a batch tag \(t\) is proposed in L1, and before it consolidates,
any agent can initiate a data availability \<fp> game.
If the staking agents do not respond or only respond with data that
does not correspond with the compressed version of \(b\), \Att can use
the compress-and-hash \<fpm> to discard batch tag \(t\).
Otherwise, at least one staking agent posts the compressed version of
\(b\) correctly signed, and \Att (and anyone observing the L1) can
decompress the data offchain and obtain \(b\).
Hence, either \Att learns the transaction requests in batch tag \(t\)
or \(t\) is discarded.

\begin{restatable}{proposition}{DAC} \label{prop:data-availability}
  For each not yet consolidated signed batch tag posted in L1, any
  honest agent can force to either learn its transaction requests or
  get the batch tag removed.
\end{restatable}

Conversely, honest arranger replicas only stake in batch tag for which
they know its batch of transaction requests and have collected at
least \(\s\) signatures of the compressed version of the batch.
Therefore, if an honest arranger replica \Rtt stake in batch tag
\(t\), it can always win any data availability \<fpm> against \(t\).
In particular, as response to the initiation of the data availability
\<fpm> against \(t\), \Rtt posts in L1 the compressed version of
the batch associated with \(t\).
The challenging agent can either not respond, in which case \Rtt (and
all agents staking on \(t\)) win the \<fp>, or proceed to the
compress-and-hash \<fpm> which, by corollary~\ref{corollary:rdoc},
\Rtt can win as the data posted is correct.

%All staking agents participate in this \<fpm> simultaneously and
Even when the data is revealed, some staking agents still lose some
tokens due to the cost of publishing data in L1.
To encourage replicas to participate and avoid silent
replicas,\footnote{Free-rider replicas.} all replicas that stake in a
batch tag also put a ``communal stake''.
For each challenge, part of this communal stake is removed.
Paying ahead of time for the cost of posting the challenge response in
L1, ensures that the replicas responding to the challenge do not incur
in a higher cost than silent replicas.

\subsection{Membership Fraud-proof
  mechanism}\label{sec:fraudproofs:membership}

Here we explain the \emph{membership \<fpm>}, which is used in some of
the \<fps> designed to prevent the consolidation of illegal batch
tags, and prove violations of property \PrLegality.

In the membership \<fpm> an agent \Att claims that element \(e\) is
the \(i\)-th leaf in a Merkle tree \(mt\) whose root is \(h\)---where
$h$ is the hash known in L1---while another agent \Btt denies the
claim.
We assume that both \Att and \Btt know all nodes in the Merkle
tree \(mt\).
We describe here two versions of this \<fpm>, one that consists of
only one (expensive) step, and another that consists in number of
simple steps, with a total number of steps that is logarithmic on the
height of the Merkle tree \(mt\).
\hypertarget{multistep}{}\addtextrevision{The multi-step membership version can be used when the
  one-step version requires too much gas due to the size of the batch.
  That is, the multi-step membership mechanism allows to split the
  one-step version in multiple smaller steps, where each step involves
  a cheaper interaction with the L1.
  For simplicity, we only explain the multiple-step version as a
  bisection.
  This can be generalized to a k-dissection as long as each
  step does not exceed the maximum gas limit of a transaction.}

In the one-step version, \Att posts in a L1 contract, called
\emph{one-step membership contract}, an element \(e\), position \(i\),
and a sequence of hashes \(\pi\).
The one-step membership contract checks that \(\pi\) is the
membership-proof that encodes that element \(e\) is at position \(i\)
of the Merkle tree with root \(h\).
The check consists in computing an amount of hashes equal to the
height of Merkle tree \(mt\).
If the proof is valid, \Att wins the game, otherwise \Btt wins the
game.
It is easy to see that \Att has a winning strategy if and only if
\(e\) is the \(i\)-th leaf in Merkle tree \(mt\).

In the multiple step version, \Att posts in an L1 contract called
\emph{multi-step membership contract} an element \(e\), its hash
\(h_e\), position \(i\), and a hash \(h_m\), claiming not only that
\(e\) is the element in the \(i\)-th leaf of Merkle tree \(mt\) but
also that \(h_m\) is the hash in the node in the middle of the path
$\pi$ from the \(i\)-th leaf to the root in Merkle tree \(mt\).
The multi-step membership contract first checks that $e$ hashes to
\(h_e\), and is also used to arbitrate the bisection game over $\pi$.
Agents alternate performing the following operations: \Btt selects one
sub-path to challenge, and if the sub-path is longer than two, \Att
responds by providing the middle hash of the sub-path chosen by \Btt.
When a sub-path has only two nodes, where the parent has hash $h_p$ and
the the child has $h_c$, \Att must post a hash \(h_s\).
The multi-step membership contract then verifies that hashing the
corresponding (left or right) concatenation of \(h_s\) with \(h_c\)
results in \(h_p\).
If the hashes are correct, \Att wins the game, otherwise \Btt wins the
game.
The multi-step membership contract verifies hashes only during \Att's
first and last turn
%and when a sub-path is of length two,
and thus, each
invocation to the contract is cheap.
Since the path length is halved at each turn, the maximum number of
turns is logarithmic in the path length, which is already logarithmic
in the number of elements in the Merkle tree \(mt\).
%
% \addtextrevision{Fig.~\ref{fig:membership-multistep} presents the
%   different states of the multi-step membership \<fp>, the allowed
%   moves in each state, and the data required for each move.}

\imarga{TODO: figure like in equilibrium}

\subsubsection*{Honest strategies for the multi-step membership \<fpm>}

If \Att's claim is correct, then \Att can win the multi-step
membership challenge, as it knows all nodes in the path from \(e\) to
the root, and all the neighbors of all nodes in the path.
Therefore, \Att can provide all the hashes requested when bisecting
the path and the hashes in the final step.

Conversely, if \Att's claims is incorrect, \Btt can always keep the
invariant that the top node in the challenged path is part of the path
from the \(i\)-th leaf to the root of \(mt\), but not the bottom node.
This can be achieved by choosing the top path when the hash proposed
by \Att in its previous turn does not match the hash of the corresponding
node in the Merkle tree \(mt\), and the bottom path otherwise.
Eventually, when the challenged path has only two elements, \(h_p\)
and \(h_c\), \Att cannot provide a hash \(h_s\) such that hashing the
corresponding (left or right) concatenation of \(h_s\) with \(h_c\)
results in \(h_p\).
If \Att were able to provide such a hash this would mean that \(h_c\)
is the child of \(h_p\) but since \(h_p\) is in the path \(h_c\)
would also be in the path.

\begin{proposition}\label{prop:membership}
  Let \(mt\) be a Merkle tree with root \(h\) known in L1, and let
  agents \Att and \Btt know all elements in \(mt\).
  Assume \Att has initiated a membership \<fp> game claiming that
  \(e\) is the \(i\)-th leaf in \(mt\).
  If \(e\) is the \(i\)-th leaf in \(mt\), then \Att can win the game.
  Otherwise, \Btt can win.
\end{proposition}

\subsection{Fraud-proofs to Prevent Illegal  Batch Tags}\label{sec:incentives:challenges}
We now present \<fpms> to prevent the consolidation of illegal batch
tags.
We assume that agents know the original batch of transactions (see
Section~\ref{sec:fraudproofs:data-availability} or
Section~\ref{sec:incentives:translate}).
If a batch tag $t=(\<batchId>,h,\sigma)$ is illegal, depending on
which condition is violated agent \Att can play one of the following
\<fps>.

\textit{\textbf{Certifiability \<fpm>}}, violation of condition
\PrCertified (\(\sigma\) does not contain \(\s\) valid signatures).
A function in L1 smart contract that arbitrates certifiability checks
whether $\sigma$ has enough signatures, and another function in the
same contract computes the aggregated public key of the claimed
signers and checks that the signature \(\sigma\) is
correct.\footnote{Alternatively, a \<fpm>, arbitrated by an L1
  contract, can arbitrate that the algorithm that checks
  multi-signatures accepts the signature provided.}
If \Att invokes one of these functions and \(\sigma\) is found to be
incorrect, then the batch tag \(t\) is discarded.  \hypertarget{removed}{}
\removetextrevision{, and \Att is rewarded for exposing the fraud.}
  
The replica that posted \(t\), as well as all agents who explicitly
staked on $t$, lose their stakes.
However, individual signers do not lose their stakes, as it is the
poster responsibility to ensure that the batch tag has enough valid
signatures and are all correct.
As described, this is a one-step \<fp> game.

\textit{\textbf{Validity \<fpm>}}, violation of condition \PrValidity.
Agent \Att claims that there is an invalid transaction request $e$ in
the batch $b$ of tag \(t\).
This \<fpm> is similar to the membership \<fp> explained before
except that in the first step the L1 contract verifies that \(e\) is
an invalid transaction request.

\textit{\textbf{Integrity \<fpm> 1}}, violation of condition
\PrIntegrityOne.
Agent \Att claims that a transaction request \(e\) appears twice in
\(t\).
\addtextrevision{Stakers can challenge any of the occurrences of \(e\), by
  selecting a path to challenge and initiating a membership \<fp>.
\Att wins the Integrity \<fpm> 1 only if \Att wins a membership \<fpm>
against all of the stakers.}\marga{not referenced in table}
% This \<fp> consists in playing two membership \<fpms>.
% %
% \Att wins the Integrity \<fpm> 1 only if \Att wins both membership
% \<fpms>.

\textit{\textbf{Integrity \<fpm> 2}}, violation of condition
\PrIntegrityTwo.
Agent \Att claims that a transaction request \(e\) in \(t\) appears in
a previous legal batch tag \(t'\).
As in the previous case, \addtextrevision{stakers can challenge any of
  the occurrences of \(e\),
  initiating a membership \<fp>, and \Att is required to win all of the FPs in order
to win the  Integrity \<fpm> 2.}

% this \<fp> consists in playing two membership
% \<fpms>, and \Att is required to win both of them in order to win the
% Integrity \<fpm> 2.

\subsubsection*{Honest Strategies}

An agent \Att that knows the transaction requests in an illegal batch
tag \(t\) and all previous batch tags, can play the \<fp>
corresponding to condition violated, win it and force the batch tag to
be discarded.

\begin{restatable}{proposition}{legality}\label{prop:legality}
  Let \(t\) be an illegal batch tag not yet consolidated in L1, and
  let \Att be an agent that knows the batch of \(t\) and all previous
  legal batch tags.
  \Att can prevent the consolidation of \(t\).
\end{restatable}

Conversely, an honest agent staking in a legal batch tag \(t\) can win
all Certifiability, Validity, Integrity 1 and Integrity 2 \<fpm>
initiated against \(t\).

\subsection{Fraud-proof Mechanism for Property Unique Batch}

The \<fpms> described so far can be used to enforce that all batches
that consolidate are legal and available, and to prove in L1 when the
arranger violates \PrAvailability and \PrLegality.
However, these mechanisms are not enough to prove in L1 violations of
the \PrUniqueBatch property, which would allow arrangers to propose
batch tags that fork the evolution of the L2.

To prevent this kind of attacks from malicious arrangers, we present a
simple \<fpm>, called \emph{unique batch} \<fpm>, which consists of
only one step: an agent \Att invokes the unique batch L1 contract
submitting two batch tags: \((\<batchId>_1,h_1,\sigma_1)\) and
\((\<batchId>_2,h_2,\sigma_2)\), such that at least one of them has
not consolidated yet.
The unique batch contract verifies that:
\begin{itemize}
  \item both batch tags have the same
    identifier, \(\<batchId>_1 = \<batchId>_2\),
  \item each batch tag is certified, and both \(\sigma_1\) and
    \(\sigma_2\) have at least \(\s{}\) arranger replicas signatures
    each, and
  \item the hashes do not match, \(h_1 \neq h_2\).
\end{itemize}
In this case all staking agents lose their stake, agent \Att is
rewarded, and it is established that all arranger replicas whose
signature appear in \(\sigma_1\) or \(\sigma_2\) should to be removed
from the arranger, as property \PrUniqueBatch is violated.
A protocol to replace replicas is outside the scope of this paper.
However, it is not possible to identify the specific replicas that are
faulty (the signers of $\sigma_1$, the signers of $\sigma_2$ or both).
This limitation follows from the inability to distinguish between the
batches agreed by consensus and batches created by faulty replicas.

\subsubsection*{Honest Strategies}

It is easy to see that if an arranger violates property
\PrUniqueBatch, any honest agent can use the unique batch \<fpm> to
prove the violation in L1 and claim a reward.

\begin{restatable}{proposition}{uniqueness}\label{prop:uniqueness}
  Violations to \PrUniqueBatch can be proven in L1.
\end{restatable}

% \begin{proof}
%   An arranger violates property \PrUniqueBatch by posting two
%   certified batch tags, \(t_1\) and \(t_2\), with the same identifier
%   but different hash.
%   % 
%   In this case, before the second batch tag consolidates, any agent
%   can use the unique batch \<fpm> by submitting \(t_1\) and \(t_2\) to
%   the unique batch contract.
%   %
%   The batches passes the validation check performed by the L1
%   contract.
%   %
%   Consequently, the violation of property \PrUniqueBatch is proved in
%   L1, and it is established that the replicas involved should be
%   replaced.
% \end{proof}

Conversely, if an arranger does not violate property \PrUniqueBatch,
then all invocations to unique batch contract fail.

\subsection{Mechanization using Lean4}\label{sec:lean}
\newcommand{\eqdef}{\ensuremath{\mathop{\stackrel{\text{def}}{=}}}}

% The goal of this section is to describe the insights gained using a
% mechanization in LEAN4 of the results hin this section, concluding in
% proving the main correctness claim:
In general, in order to prove that a single honest
player can prevent the consolidation of illegal and unavailable batch
tags and can prove in L1 violations to safety properties, one has to
prove all \<fpms> presented in the previous sections.
However, in LEAN4, we focus on proving that only legal batches
consolidate, as the correctness of the data-availability \<fp> follows
from the correctness of RDoC, and the unique batch \<fpm> is a
one-step \<fp> where the unique batch L1 contract perform a simple
verification to decide the winner.
\removetextrevision{
Also, we do not model time and assume players can interact freely with
the L1, i.e. we only focus on the correctness of the arguments.
}\hypertarget{fp:time}{}\addtextrevision{We focus on the correctness arguments, so we do not
  model time in our proofs, assuming that players can interact freely
  with the L1.
%
% Modeling time, unless we have a realistic model of the blockchain ecosystem, is
% of no use.
%
  Similarly to the FPs for STFs in optimistic rollups, we assume that
  agents have enough time to perform actions.
}
We define claims about batch tags legality as disputable assertions
(DAs) composed of a tree where elements are at the leaves, and such
that the root of the Merkle tree is the hash of the tag.

From Section~\ref{sec:incentives:challenges}, we have four different
\<fpms>, one for each way of making a batch illegal.
Certifiability \<fpm> does not require any LEAN check as is a one-step
\<fp> which is directly checked by a L1 smart contract.
\<fpms> Validity and Integrity 1 are local to the batch tag, while
Integrity 2 depends on the past history of the L2 batch tags.
% %
\removetextrevision{ We focus on the local notion of validity, we are
  working on proving the history dependent property Integrity 2 and we
  leave it as future work.
}\hypertarget{fp:lean}{}\addtextrevision{We define two notions of
  valid DA, one is local and depends only on the values in a single
  batch, while the other depends on the entire history.  }
\martin{Forgot nomenclature here. What's the proper way of calling the
  history?}

\begin{definition}[Local Valid DA]
  We say that a batch of transaction requests is local valid if and
  only if all its transaction requests are valid and there are no
  duplicated elements.
\end{definition}
\hypertarget{fp:leanDef}{}
\addtextrevision{
  \begin{definition}[History Valid DA]
    We say that a batch of transactions requests is history valid if and only if
    it is local valid and no transaction appears in the history.
  \end{definition}
}

In LEAN, we write the above definitions as, assuming that
\(\mathsf{validTr}\) is a predicate indicating when a transaction
request is valid:

% \removetextrevision{
%   \[
%     \mathsf{localValid}({tree}, h)
%     \eqdef
%     \begin{array}{rl}
%       & \mathsf{merkleRoot}({tree}) = h \\
%       \wedge & \forall e \in {tree}, \mathsf{validTr}(e) \\
%       \wedge & \mathsf{NoDup}({tree}) \\
%     \end{array}
%   \]
% }  

  \begin{align*}
    \mathsf{localValid}(\mathsf{tree}, h)
    \eqdef &
    \begin{array}{rl}
      & \mathsf{merkleRoot}(\mathsf{tree}) = h \\
      \wedge & \forall e \in \mathsf{tree}, \mathsf{validTr}(e) \\
      \wedge & \mathsf{NoDup}(\mathsf{tree}) \\
    \end{array} \\
\hypertarget{fp:Lean:FormalDef}{}
%\addtextrevision{
    \mathsf{globalValid}(\mathsf{tree}, h, \mathsf{history})
    \eqdef &
    \begin{array}{rl}
      & \mathsf{localValid} \\
      \wedge & \forall e \in \mathsf{tree}, e \notin \mathsf{history} \\
    \end{array}%}    
  \end{align*}
The predicate \(\mathsf{NoDup}\) simply flattens the tree and checks
that there are no duplicated elements in the resulting list.
Additionally to the previous valid conditions, we have that the
proposed Merkle root hash is in fact the Merkle root of the proposed
tree.

Without getting into implementation details, we define\removetextrevision{
  a simple protocol, called
  \(\mathsf{LinearL2Protocol}\),}\hypertarget{fp:lean:protocols}{}
\addtextrevision{two simple protocols,
called \(\mathsf{LinearL2Protocol}\) and \(\mathsf{LinearHistoryL2Protocol}\),}
where one player
proposes a DA, a tree of elements plus its (supposed) Merkle root, and
the other player decides what to do next.
\hypertarget{fp:lean:protocols-diff}{}\addtextrevision{The main
  difference between the protocols is that the protocol
\(\mathsf{LinearHistoryL2Protocol}\) takes into the account the history of the
L2 blockchain.}
The options of the second player are to accept the DA or challenge one
of the claims above.
Challenging a claim means that one \<fpm> is triggered.
We capture the honesty of the second player and prove that
only valid DAs are accepted and all faulty ones are challenged and
discarded.
% \marga{we also have honest strategies for the first player, no?}

The honest second player follows a simple algorithm: first, compute
the Merkle root and if it differs from the hash proposed, trigger the
decompress-and-hash \<fpm>.\footnote{The
  decompress-and-hash \<fp> takes a compressed batch and the (supposed) Merkle
root, but in our LEAN formalization, we abstract this away for simplicity. The
decompression step can be modeled as a partial function triggering an external
one-step check.}
%
% \martin{Check it with marga how it is done
%   now.}\marga{now the data-availability fp require replicas to post
%   the batch compressed and then challengers can play the
%   ``decompress-and-hash'' fp when the data posted by replicas is not
%   correct. Should we have another function like
%   localValidCompress(compressed, h) = decompress(compressed) = tree
%   and localValid(tree,h) ?  }
% \martin{Yes, if you want to, compression is an isomorphism and I
% abstracted it away in the formalization. We can add extra steps to verify that
% the results are the same as one-shoot, exactly as you mention checking that the
% result is the expected result.}
% marga: added a footnote trying to clarify this.
%
If the proposed Merkle root is correct, the honest player checks if
all elements are valid.
In LEAN4, we find the first invalid element in the tree (if there is
one), in which case the honest second player triggers a membership
\<fpm>.
If all elements are valid, the honest player checks if there are
duplicated elements, again obtaining the first repeated element and
triggering the Integrity 1 \<fpm>, which consists of \emph{two}
membership \<fps> in sequence.
\hypertarget{fp:lean:honest}{}\addtextrevision{In the case of the honest player for the history protocol, we also check that
all elements are not present in the current history.
}

% \removetextrevision{
% As mentioned earlier in this section, we have two main \<fps>, one for
% data availability and one for membership.
% }
% \addtextrevision{
We have two main \<fps>: one for data availability and one for membership.
%}
%
In LEAN, since we care about how computations are performed, we
defined several similar \<fpm>, in particular for membership.
The multi-step membership \<fp> can be played following the path in
the Merkle tree (which is \emph{linear}), one level at a time, either
from the root to the leave or from the leave to the root.
This game can be played by bisecting the path in a \emph{logarithmic}
game, as presented in Section~\ref{sec:fraudproofs:membership}.
All these variant are implemented in the library and we are working on proving
them equivalent for honest players.
%
% However, for simplicity, we used their linear counterparts in this
% model.

Honest players in our LEAN implementation are defined as players that
know all information, in particular they know all transaction requests
in a batch.
Using this information, they can compute all correct answers to all
possible challenges when they follow the algorithm described above to
trigger a \<fp> game and then play the game following the
corresponding winning strategy.
% as the one just
% described.
% \marga{which ones?}\martin{Algorithms? Answers? I am also confused by
%   the wording. The idea is that with all the data from batches, agents
%   can compute everything and play to win all possible challenges (if
%   they want to)}
%
In the library, we proved that honest players always win the
data-availability and membership \<fpms>.
When honest players challenge faulty proposers, they win, leading to the
proposed DA be rejected.
When honest players propose a new DA, they can defend their claim and win
against all possible challengers.

\hypertarget{fp:lean:theorem}{}
\begin{theorem}[Only Valid \removetextrevision{Local} DAs]
\removetextrevision{Let \(\mathsf{Honest}\) be an honest player following the previous algorithm,
 triggering \<fpm> when needed.
 Then, the protocol \(\mathsf{LinearL2Protocol}\) with one player
 being \(\mathsf{Honest}\) only accepts local valid DAs \(\mathsf{localValid}\).}
\addtextrevision{
  Let \(\mathsf{Honest}\) and \(\mathsf{HistHonest}\) be honest players
  following the previous algorithms, triggering FPs when needed.
  Then,  \(\mathsf{LinearL2Protocol}\) and
\(\mathsf{LinearHistoryL2Protocol}\) with one player being \(\mathsf{Honest}\)
only accepts local valid DAs \(\mathsf{localValid}\) and
history valid DAs \(\mathsf{globalValid}\) respectively.
}
 \end{theorem}

In LEAN the previous result is expressed as follows, where
\((\mathsf{tree}, h)\) is a DA proposed by an arbitrary player:
\[
  \mathsf{LinearL2Protocol} \, (\mathsf{tree}, h) \, \mathsf{Honest}
  \iff \mathsf{LocalValid}(\mathsf{tree}, h)
\]
\hypertarget{fp:lean:th:formalization}{}
\addtextrevision{
and given history \(H\):
\[\begin{array}{l}
  \mathsf{LinearHistoryL2Protocol} \, (H,\mathsf{tree}, h) \, \mathsf{HistHonest}
  \iff\\ \hspace{15em} \mathsf{glovalValid}(\mathsf{tree}, h, H)
\end{array}\]
}
%
% The proof is large and tedious, although straightforward.
%
The proof relies on verifying each function precisely.
For example, when finding a duplicated element, the proof reduces to
building the proper evidence to guarantee winning the corresponding
\<fpm>.
We also need to prove that if there are duplicates the corresponding
game is triggered, ensured by the moves of our honest player.
In total, the library has approximately 5256 lines of commented LEAN
code.
The main results can be found in the file ``L2.lean'' of the companion
artifact.

%%% Local Variables:
%%% TeX-master: "main.tex"
%%% TeX-PDF-mode: t
%%% End:

\section{Incentives}
\label{sec:incentives}

In this section, we present a collection of incentives that provides
payments to replicas based on evidence of work, in particular for
generating and signing hashes, posting batches into L1, and
translating hashes into batches.
The \<fpms> from the previous section are deterrents against incorrect
behavior that result in losing the stake.
These are last-resort mechanisms to catch faulty behavior, like
illegal and unavailable batch tags posted in L1, showing evidence of
replicas faulty behavior.

We present now a framework where \emph{rational} behavior (maximizing
profit) aligns with honest behavior, i.e. rational behavior
corresponds to the the arranger protocol.
We also present a simple cost and reward analysis to determine the
minimum budget required to guarantee that no illegal or unavailable
batch tag consolidates, and to prove that an arranger violates a
safety property.
\hypertarget{treatment-incentives}{}\addtextrevision{We give here an initial
  understanding of incentives, and introduce the necessary parameters
  for instantiating a formal game-theoretic model such as~\cite{alba2025scaffino}.}
\subsection{Incentives to Generate and Post Batches}
\label{sec:incentives:generate}
Each consolidated batch generates the following rewards:

\begin{compactitem}
\item Every replica receives a constant amount $k_1$ of L2 tokens.
\item Replicas signing consolidated batches receive $k_2$ additional
  L2 tokens where $k_2$ is a monotonically increasing function on the
  number of signatures and transaction requests.
\item The replica posting the batch in L1 receives an additional
  payment of $k_3$ L2 tokens to cover the posting costs.
\end{compactitem}

These rewards create incentives to (1) participate in the arranger
protocol; (2) sign batches, and communicate and collect signatures;
(3) include as many transaction requests as possible in a batch; and
(4) post batch tags in L1.

These incentives combined with our \<fps> make deviating from the
arranger protocol irrational when there are no external payments.
As result, our system is \emph{incentive
  compatible}~\cite{Ledyard1989}. 
A more detailed quantitative study including expected utility and a
formal game-theoretic security analysis (e.g., using frameworks
like~\cite{brugger2023CheckMate}), is left as future work.

All payments are charged to L2 users submitting transaction requests
and made effective by STFs when computing the effects of batches that
consolidated, including the payments as part of new L2 blocks.
Optionally, L2 tokens can be minted by the L1 contracts that govern
the evolution of the L2.

\subsection{Incentive to Translate Hashes into Batches}
\label{sec:incentives:translate}
% The incentives above do not include the translation of hashes into
% batches (other than STFs producing L2 blocks and generating payments).
%

We present now an alternative cheaper protocol for replicas to
translate hashes and get paid using zero-knowledge contingent
payments~\cite{campanelli2017zkcontingent,fuchsbauer2019WIIsNotEnough,nguyen2020WIIsAlmostEnough,GomezMartinez25AZKCP}.
Unlike the data availability \<fpm> from
Section~\ref{sec:fraudproofs:data-availability}, this protocol
operates primarily off-chain although it does not provide evidence of
fraud if replicas remain silent.

Consider a client \Ctt, such as an STF, who wants to know the batch of
transaction requests associated with a batch tag with hash $h$ in
order to compute the next L2 block, or to check the validity of an L2
block posted.
When \Ctt requests the translation of $h$ it can contact directly
replica \Rtt, which knows the inverse translation $b$, but in this
protocol \Rtt does not respond immediately with $b$.
Instead, 
\removetextrevision{\Rtt encrypts $b$ using a secret key $k$, such that
$w = \texttt{Enc}_k(b)$ and computes $y$ such that
$\texttt{SHA256}(k) = y$.
Then, \Rtt sends $w$ and $y$ to \Ctt along with a zero-knowledge proof
that $w$ is an encryption of $b$ under key $k$ and that
$\texttt{SHA256}(k) = y$.}
\hypertarget{com:test}{}
\addtextrevision{\Rtt creates a fresh secret key $k$, and computes 
$y = \texttt{H}(k)$, where $\texttt{H}$ is an algebraic instance of the hash function $\texttt{H}$, such as $\texttt{H}(k) = g^k$, where $g$ is a generator of a cyclic group. 
Subsequently, \Rtt encrypts $b$ under public key $y$ using an algebraic encryption scheme (e.g., ElGamal~\cite{Elgamal85}). 
The observation in~\cite{GomezMartinez25AZKCP} is that it is possible to use an 
efficient Sigma protocol to create a zero-knowledge proof $\pi$ proving that $w$ is a valid encryption of $b$ under $y$ and that $y = \texttt{H}(k)$, without revealing the secrets $b$ or $k$. 
Finally, \Rtt sends $\pi$, $w$ and $y$ to \Ctt. }
Client \Ctt verifies the proof and, if it is correct, \Ctt generates a
contingent payment where \Rtt is the only beneficiary.
%
%\footnote{To
%  save time \Ctt can interact with several servers simultaneously and
%  generate a contingent payment with several beneficiaries but only
%  the first one the reveal the secret collect the
%  payment. }}
%
Replica \Rtt can collect the payment only by revealing the secret $k$.
When \Rtt reveals $k$ to collect the payment \Ctt learns $k$ so it can
decipher $w$ and obtain $b$.
Appendix~\ref{app:htlc} shows the pseudo-code implementing a payment
system based on hashed-timelock contracts.
All communication between \Rtt and \Ctt is offchain and one-to-one.
The only communication with L1 occurs when \Ctt creates the contingent
payment and when \Rtt reveals the key \(k\) in L1 to collect the
payment.
Moreover, batch \(b\) is never posted in L1.
Therefore, latency and bandwidth between replicas and clients are
likely to not represent a major bottleneck.

A faulty replica \Rtt can participate in the first part of the
protocol and then refuse to reveal $k$ causing client \Ctt to incur in
unnecessary costs.
To prevent this, \Rtt must also sign $y$ and send the signature to
\Ctt, so \Ctt has evidence.
If \Rtt fails to disclose $k$ after some time has passed, \Ctt can use
the signature as evidence to accuse \Rtt of remaining silent, have
\Rtt stake removed and receive some compensation.

In contrast to the data availability \<fpm>, this protocol does not
offer any guarantee to clients when replicas refuse to participate.
However, this protocol is cheaper for clients and more rewarding for
replicas than the data availability \<fp>.\hypertarget{reward}{}
\addtextrevision{Arranger replicas get a payment when translating
  batches using the offchain protocol but they do not receive any
  reward when the data availability challenge is used.  On the other
  hand, the offchain protocol is cheaper for clients than the data
  availability challenge.}
Moreover, if replicas reveal the data when forced to participate in a
data availability \<fp> they lose not only part of the communal stake
but also the reward that they would have obtained with this contingent
payment.
Hence, in practice, the data availability \<fp> is a deterrent
for replicas withholding data, and thus, this offchain protocol is always used
to translate hashes (see Section~\ref{sec:incentives:costs}), as both
clients and replicas benefit from it.

\subsection{Cost Analysis}
\label{sec:incentives:cost-analysis}
\label{sec:incentives:costs}
We perform a simple cost and reward analysis for each \<fp> mechanism
presented in Section~\ref{sec:fraudproofs} and the translation
protocol (see Section~\ref{sec:incentives:translate}) to determine the
minimum budget required by honest agents to guarantee that no illegal
or unavailable batch consolidates, and to prove violations of safety
properties.

First, we define the following variables where $x$ can be \KWD{data},
\KWD{certifiability}, \KWD{validity}, \KWD{integrity1},
\KWD{integrity2} or \KWD{uniqueness}, \addtextrevision{and $y$ can be
  a move in a \<fpm>:}
% \KWD{init-data},
% \KWD{init-certifiability}, \KWD{init-validity}, \KWD{init-integrity1},
% \KWD{init-}:
\begin{compactitem}
\item $s$ represents the tokens needed to stake in a batch tag.
\item
  $s_{\KWD{\addtextrevision{com\_}data}}$ are the tokens
  provided in a communal stake by each staker.
\item \addtextrevision{$s_{x}$ represents the stake needed by a client
    to start \<fp> $x$.}
\item $\CC{\KWD{translate}}$ is the client's costs for using the offchain
  translation protocol. \addtextrevision{This includes the gas used to
    create the contract and the payment done to the staker.}
\item
  $\SC{\KWD{translate}}$ and $\SR{\KWD{translate}}$ are
  stakers cost (resp. reward) for
  translating a hash
  using the offchain protocol.
\item $\CC{x}$ and $\SC{x}$ are client (resp. stakers) cost for
  participating in the \<fpm> of type $x$.
  \item
  $\CR{x}$ are client reward for winning the
  \<fp> of type $x$.
  \hypertarget{concrete-values:new-vars}{}
\item \addtextrevision{$C_{y}$ represents the cost of performing move
    $y$. The cost is covered by the player performing the move, which
    depending on move $y$ can be the staker or the client.}
\end{compactitem}
All of these variables are in L1 tokens.
\newcommand{\SZ}{\KWD{SZ}}

\addtextrevision{
The cost of performing a move represents the gas required, which
depends on the underlying L1 blockchain.}
The cost for participating in an \<fpm> \removetextrevision{represents
  the gas needed in the worst case}\addtextrevision{is upper bounded
  by the gas needed in the worst case, i.e., the sum of the cost of
  all moves in the worst case.}
% \addtextrevision{which can be expressed as a function of the cost of
% performing its moves, and in, some cases, the size of the batch, which
% we denote with \SZ}.
\removetextrevision{which depends on the underlying L1 blockchain.
 The cost of staking should be adjusted taking in account the cost of
 playing the \<fpm> and also the value locked in the L2.}

Stakers do not receive any reward for winning an \<fp> game, as their
motivation for participating is to not lose their stake and eventually
collect their reward from consolidated batches.
However, they do receive a reward for correctly translating hashes
using the offchain protocol as there is no stake at risk there.

The following relations must hold between the parameters:
\begin{compactenum}
\item \label{client-cost-reward}
  The client costs to participate in \<fpm> $x$ is
  significantly less than the reward for winning it,
  $\CC{x} \ll \CR{x}$.
\item \label{staker-cost-stakes} Clients cover the costs for stakers in all \<fpms> except in the
  data availability \<fpm> where this cost is covered by the communal
  stake.
  In symbols, \(\sum s_{\KWD{\addtextrevision{com\_}data}} > \SC{\KWD{data}}\) and for all
  remaining \<fpm> \(x\), \removetextrevision{$\CC{x} > \SC{x}$} \addtextrevision{$s_{x} > \SC{x}$.}
  \removetextrevision{
In particular, no staker plays the certifiability and
        unique-batch \<fps>, therefore
        $\SC{\KWD{certifiability}} = \SC{\KWD{uniqueness}}= 0$.}
  \item \label{client-reward-staker-stake}
    The stakes taken from the losing player cover the rewards.
    For the validity and integrity \<fpms>, the client reward for
    winning the \<fp> is less than the challenged agent stake,
    $\CR{x} < s$.
    For the data, certifiability and unique batch \<fpms>, the client
    reward is bounded by the sum of all stakes in the batch,
    i.e. $\CR{\KWD{certifiability}} < \sum s$, $\CR{\KWD{data}} < \sum s$ and
    $\CR{\KWD{uniqueness}} < \sum s$.
  For the offchain translation protocol, clients cover the
  stakers reward, $\CC{\KWD{translate}} > \SR{\KWD{translate}}$.
\item \label{offchain-replica}The replicas must have a motivation to
  engage in the offchain translation protocol and gain a profit.
  That is, the staker costs to translate hashes in the offchain
  protocol is significantly less than its reward,
  $\SC{\KWD{translate}} \ll \SR{\KWD{translate}}$.
\item \label{offchain-client} Clients must have a motivation to use offchain translation.
  That is, the client cost for using the offchain translation protocol
  is significant less than participating in the data availability
  \<fpm>, $\CC{\KWD{translate}} \ll \CC{\KWD{data}} + s_{\KWD{data}}$.
  \hypertarget{concrete-values:new-relations}{}
  \addtextrevision{\item \label{client-cost-data} In the worst case,
    during the data-availability \<fp> the client must initiate the
    game and bisect the execution trace of the WASM compilation of
    program P until the traces is reduced to a single instruction.
    Therefore, the client cost for participating in the data-availability
    \<fp> is
    $\CC{\KWD{data}} = C_{\KWD{init\_data}} + l \times
    C_{\KWD{bisect\_subtrace}}$, where \(l\) is the length of the
    challenged trace.}
\addtextrevision{\item The client cost for
    participating in the certifiability \<fp> is the maximum cost
    between checking the number of signatures and checking
    the aggregated signature, that is,
    $\CC{\KWD{certifiability}} = \max(C_{\KWD{check\_size}},
    C_{\KWD{check\_agg}})$.}  \addtextrevision{\item Unique-batch
    \<fp> is a one-step game, then
    $\CC{\KWD{uniqueness}} = C_{\KWD{unique\_batch}}$}
  \addtextrevision{\item \label{client-cost-integrity}In the worst-case, during validity, integrity
    1 and integrity 2 \<fps> a client must initiate the game, bisect
    the path until it consists of only two nodes, and reveal the sibling
    of the lower node. In symbols, the client cost for participating in
    the validity \<fp> is
    $\CC{\KWD{validity}} = C_{\KWD{init\_validity}} + (\log\log(\SZ) -
    1) \times C_{\KWD{bisect\_subpath}} + C_{\KWD{reveal\_sibling}}$,
    $\CC{\KWD{integrity_1}} = C_{\KWD{init\_integrity1}} +
    (\log\log(\SZ) - 1) \times C_{\KWD{bisect\_subpath}} +
    C_{\KWD{reveal\_sibling}}$,
    $\CC{\KWD{integrity_2}} = C_{\KWD{init\_integrity2}} +
    (\log\log(\SZ) - 1) \times C_{\KWD{bisect\_subpath}} +
    C_{\KWD{reveal\_sibling}}$. }
  \addtextrevision{\item \label{staker-cost-data} In the
    data-availability \<fp> the staker must provide the compressed
    batch and then select which subtrace to challenge until the
    subtrace correspond to a single instruction.
    Therefore, stakers cost for participating in the data-availability
    \<fp> is
    $\SC{\KWD{data}} = C_{\KWD{post\_compressed}} + (l-1) \times
    C_{\KWD{select\_subtrace}}$.}  \addtextrevision{\item No staker
    plays the certifiability and unique-batch \<fps>, therefore
    $\SC{\KWD{certifiability}} = \SC{\KWD{uniqueness}}= 0$.}
  \addtextrevision{\item In the validity \<fp> the staker must select
    which subpath to challenge until the subpath has length 1. Then,
    stakers cost for participating in the validity \<fp> is
    $\SC{\KWD{validity}} = \log\log(\SZ) \times
    C_{\KWD{select\_subpath}}$.}
  \addtextrevision{\item \label{staker-cost-integrity}In the integrity
    1 and 2 \<fps> the staker must first select which path to
    challenge and then select which subpaths to challenge until the
    subpath has length 1. Therefore, stakers cost for participating in
    the integrity 1 and 2 \<fps> is
    $\SC{\KWD{integrity1}} = \SC{\KWD{integrity2}} =
    C_{\KWD{select\_path}} + \log\log(\SZ) \times
    C_{\KWD{select\_subpath}}$.}

\end{compactenum}

We now determine the minimum budget that an honest agent \Att needs to
guarantee that illegal or unavailable batch tags are discarded.
This amount is computed from the maximum requirement for playing the
corresponding mechanism.
In the case of \PrValidity, \PrIntegrityOne, and \PrIntegrityTwo, \Att
may need to get the data first, and thus, \Att needs to account tokens
for playing the data availability \<fp> game.
Agent \Att has to play these mechanisms against all stakers.
The budget proposed is enough if \Att plays all of them in sequence,
accumulating rewards along the way.

\begin{restatable}{proposition}{costs}\label{prop:cost}
  An agent \Att that knows all transaction requests in previous legal
  batch tags can discard unconfirmed illegal or unavailable batch tags,
  provided that \Att has at least the following
  \removetextrevision{tokens} \addtextrevision{budget
    \(B =
    \)}\(max(\addtextrevision{s_{\KWD{certifability}}}+\CC{\KWD{certifability}},
  \addtextrevision{s_{\KWD{data}}} + \CC{\KWD{data}} +
  max(\addtextrevision{s_{\KWD{validity}}} + \CC{\KWD{validity}},
  \addtextrevision{s_{\KWD{integrity1}}} + \CC{\KWD{integrity1}},
  \addtextrevision{s_{\KWD{integrity2}}} + \CC{\KWD{integrity2}}))\).
\end{restatable}

The complete proof can be found in
Appendix~\ref{app:proofs}. \imarga{TODO: update proof with stake}

\hypertarget{implications-prop-five}{}
\addtextrevision{
Since the required budget to discard illegal or unavailable  
batch tags, \(B\), depends on clients stake and the size of the batch,   
any L2 that decides to implement our proposal can set those values
to obtain an upper-bound for \(B\).
By limiting the number of transaction requests per batch and the stake
clients must put to play each \<fpm>, any implementation of our
proposal can guarantee that an adversary cannot construct an illegal
batch tag with arbitrary cost in such a way that no honest agent \Att
has sufficient tokens to successfully challenge its illegality.
More generally, our work provides a framework with its proof of
correctness, but leaves variables like \SZ and \(s_{x}\) open.
This allows for concrete implementations to set these parameters
according to their specific use cases and requirements.}

By discarding an illegal or unavailable batch tag that is certified,
an agent also proves that the arranger violates properties \PrLegality
and \PrAvailability.
Proving that an arranger violated \PrUniqueBatch requires playing the
unique batch \<fp>.
Therefore, proving that an arranger violates any safety property
requires the cost of playing the unique batch \<fp> and the cost of
discarding illegal and unavailable batches.

\begin{restatable}{corollary}{cost-safety}\label{prop:cost-safety}
  An agent \Att that knows all transaction requests in previous legal
  batch tags can prove that the arranger violated a safety property,
  provided that \Att has at least the following tokens
  \removetextrevision{$max(\CC{certifability}, \CC{data} +
    max(\CC{validity},\CC{integrity1},\CC{integrity2}))$}
  \addtextrevision{$max(B, s_{\KWD{uniqueness}}+\CC{\KWD{uniqueness}})$}.
\end{restatable}

Finally, the offchain protocol is cheaper for clients and replicas, and
more rewarding for replicas than the data challenge.
Therefore, rational agents are inclined to use the offchain protocol.

Furthermore, by adjusting the allocation of the total cost of the data
availability \<fpm> between clients and stakers, the protocol can
influence the equilibrium price of the contingent payment.
\removetextrevision{This will be explored in more detail in a full
  version of the paper.}\marga{not referenced in table. we have an
  appendix the explore this a bit, but i do not if it is correct}

\newcommand{\GasUsed}{100.000\xspace}
\newcommand{\GasPrice}{3\xspace}
\newcommand{\MoveCost}{0.0003\xspace}
\newcommand{\stake}{10\xspace}
\newcommand{\budget}{20.6141\xspace} %2*stake + moveCost * 47 + Cpost_data

\hypertarget{concrete-values:analysis}{}
\addtextrevision{
\subsection{Analysis with Concrete Values}
Here we provide representative values for the variables introduced
earlier, under the assumption that the underlying L1 is Ethereum and
that the batch size is bounded by \(\SZ \leq 4096\) transaction
requests.
For simplicity, we assume that each move consumes at most \GasUsed
gas, which correspond to the the median gas used in Ethereum
transactions, and a gas price of \GasPrice Gwei.\footnote{Data
  obtained from
  https://studio.glassnode.com/charts/fees.GasUsedMean?a=ETH and
  https://studio.glassnode.com/charts/fees.GasPriceMean?a=ETH,
  respectively, on July 2025.}
Except for transaction posting data compressed, which before blobs,
was 0.06 ETH.~\footnote{See, for example, transaction 
 % https://etherscan.io/tx/0x4276c81d7b5188cf9ebd57c39079e1e673323e42bc12d3ec620693e0a6bbade5,
   0x4276c81d7b51...X0693e0a6bbade5,
where the
  Arbitrum Sequencer post a batch in Ethereum.}
Under these assumptions, the cost of move \KWD{post\_compressed} is
\(C_{\KWD{post\_compressed}} = 0.6\) and the cost of all other moves
\(y\) is \(C_{y} = \MoveCost\) ETH.
Using relations \ref{client-cost-data}-\ref{client-cost-integrity},
we compute the client cost for each \<fpm> as follows:
\(\CC{\KWD{data}} = 0.6081\) ETH,
\(\CC{\KWD{certifability}} = \CC{\KWD{uniqueness}} = \MoveCost\) ETH,
\(\CC{\KWD{validity}} = \CC{\KWD{integrity1}} = \CC{\KWD{integrity2}}
= 0.0039\) ETH.
These values give a lower bound for clients rewards (see
relation~\ref{client-cost-reward}), which in turn give a lower bound
for staker stakes in the batch tag (see
relation~\ref{client-reward-staker-stake}).
Assuming clients rewards for each \<fp> \(x\) is 1 ETH plus the
client cost for \<fp> \(x\), we obtain: \(\CR{\KWD{data}} = 1.0084\) ETH, \(\CR{\KWD{certifability}} =
\CR{\KWD{uniqueness}} = 1.0003\) ETH, \(\CR{\KWD{validity}} =
\CR{\KWD{integrity1}} = \CR{\KWD{integrity2}} = 1.0039\) ETH.
Given these bounds, a stake of \stake ETH
%, as proposed in~\cite{bold-economics-of-disputes} for STFs,
is more than enough to be used as an stake for arranger replicas.}

\addtextrevision{Similarly, using relations \ref{staker-cost-data}-\ref{staker-cost-integrity} we compute the staker cost for each
  \<fp>: \(\CC{\KWD{data}} = 0.0084\) ETH,
  \(\CC{\KWD{certifability}} = \CC{\KWD{uniqueness}} = 0\) ETH,
  \(\CC{\KWD{validity}} = 0.0036\) ETH, and
  \(\CC{\KWD{integrity1}} = \CC{\KWD{integrity2}} = 0.0039\) ETH.
  By relation \ref{staker-cost-stakes}, these values give lower
  bounds for both the communal stake and clients stake in each \<fp>.
  In particular, the communal stake con be set to
  \(s_{\KWD{com\_stake}} = 0.009\) ETH, and client stake \(s_{x}\) can
  be set to \stake ETH for all \<fps>.
  %, matching the stake proposed for STFs in Arbitrum.
}
\addtextrevision{Therefore, in this scenario, by
  Proposition~\ref{prop:cost}, the minimum budget needed to discard
  illegal or unavailable batches is \(B = \budget\) }
\addtextrevision{Regarding the offchain translation protocol, replicas
  cost is the cost of a single transaction,
  \(\SC{\KWD{translate}} = \MoveCost\), while clients cost is the cost
  of single transaction plus the replica reward,
  \(\CC{\KWD{translate}} = \MoveCost + \SR{\KWD{translate}}\).
To satisfy relations \ref{offchain-replica} and
\ref{offchain-client}, replicas reward must satisfy: \(\MoveCost \ll \SR{\KWD{translate}} \ll \CC{\KWD{data}} +
s_{\KWD{data}} - \CC{\KWD{translate}} = \stake\) ETH.
Thus, a reasonable reward in this setting is 1 ETH. 
}

\newcommand{\kone}{1\xspace}
\newcommand{\ktwo}{200\xspace}
\newcommand{\kthree}{10\xspace}
\newcommand{\n}{31\xspace}
\newcommand{\fee}{0.5\xspace} %(kone * n + ktwo * n/3 + kthree)/SZ
% (31+200+10)/4096
\newcommand{\ratio}{\frac{1}{20}} %kthree/fee
\addtextrevision{Finally, when adding transactions requests L2 users
  must pay a fee enough to cover replicas rewards upon batch
  consolidation (see Section~\ref{sec:incentives:generate}).
  Suppose a batch tag \(b\) with \(S'\) signatures consolidates.
  The total reward distributed to replicas is given by:  $k_1\times n + k_2
  \times S' + k_3$.
  % which must be divided
  % between the L2 users whose transactions requests are included in
  % \(b\).
  %
  Since the concrete value of \(S'\) and the batch size are unknown at
  the time of transaction submission, we propose that each L2 user
  pays a flat fee of $(k_1\times n + k_2 \times S + k_3)/\SZ$.
  If the total collected fees are insufficient---e.g., due to more
  signatures than expected---L2 tokens can be minted to make up the
  difference, introducing inflation.
  Consider an scenario of a fully decentralized arranger consisting of
  \(n = \n\) replicas where the cost of executing a transaction in L1 corresponds
  to \kthree L2 tokens. By setting parameters $k_3 = \kthree$,
  $k_2 = \ktwo$, and $k_1 = \kone$, the resulting fee for L2 users would be approximately 
\fee.
  This is roughly \(\ratio\)-th the cost of executing a transaction on L1, providing a significant cost advantage for L2 users.
}
%%% Local Variables:
%%% TeX-master: "main.tex"
%%% TeX-PDF-mode: t
%%% End:

\section{Threat Models}
\label{sec:threat-models}

We consider three types of adversarial models, all computationally
bounded, as in~\cite{Cachin2000Random,Cachin2001Secure}.
We analyze how two of these adversary can impact the evolution of the
L2 under the different implementations of the arranger described in
Section~\ref{sec:arranger:implementations} (centralized,
semi-decentralized, and fully-decentralized).
The third adversarial model is specific to the fully-decentralized
implementation, as it reasons about the batch tags consented by honest
replicas.

\paragraph*{\textbf{\advOne Adversary}}
In this adversarial model, all replicas are either honest or Byzantine, and
we assume that the trust assumptions of the arranger implementation are
not violated.

In particular, for the centralized implementation this means that the
only replica is honest.
In the semi-decentralized implementation the sequencer and at least
\(n-\s{} + 1\) DAC members are honest.
Finally, for the fully-decentralized implementation the adversary can
control less than one third of the replicas.

Therefore, under this adversary, the arranger implementation is
correct (see~\cite{capretto2025decentralized} for a proof for the
semi-decentralized and fully decentralized implementations).
However, Byzantine replicas can still post batch tags that are not
certified, so the certifiability \<fp> should be utilized to discard
such batch tags.

\paragraph*{\textbf{\advThree Adversary}}
This is the strongest adversary, capable of controlling all arranger
replicas.
This adversary can decide the content of all legal batch tags created
and posted by the arranger.
In particular, the adversary has the power to exclude valid
transaction requests from posted legal batch tags, potentially
violating property \PrTermination.
Unfortunately, detecting that the arranger is violating \PrTermination
is impossible as this is a liveness property that might be satisfied
later in the future.
Then, this adversary can censor transactions without any consequence.
To mitigate this risk, L2 users can post their transaction requests
directly in L1 although at a higher cost.

This adversary can also violate all safety properties, but this offers
no advantage to the adversary, as this violations can be detected
using our \<fps> (see Theorem~\ref{cor:safety}), and merely exposes
the corruption of the arranger.
Furthermore, even under this extremely powerful adversary,
Lemma~\ref{l:consolidate} holds and a single honest agent can
guarantee the legality and data availability of consolidated batch
tags.

  % \item \textbf{\advTwo Adversary}:
\paragraph*{\textbf{\advTwo Adversary}}
The power of this adversary lies between the previous two and only
applies to the fully-decentralized arranger.
\hypertarget{corrupt}{}\addtextrevision{The previous two adversaries
  capture what happens with arranger implementations either when no
  trust assumption is violated or when the adversary has full control
  over the arranger.
  Here, we consider a more practical adversary where the trust
  assumptions of the arranger are violated but honest replicas still
  have the power to create certified batch tags.}

This adversarial model refines the Byzantine
failure model~\cite{Lamport1982Byzantine} by introducing a new type of
faulty replicas: \emph{corrupt} replicas.
Corrupt replicas behave arbitrarily without violating any property of
SBC, that is, they only misbehave in the DAC part of the arranger.
For example, corrupt replicas may sign incorrect batches or refuse to
translate batches, but corrupt replicas can neither prevent honest
replicas for agreeing on the content of batches nor prevent honest
replicas from including elements added by clients to honest replicas.
The portion of Byzantine replicas is less than one-third, the total
number of honest replicas is at least \(\s{} = \frac{n}{3}\), where
\(\s{}\) is the minimum number of signatures required for a batch to
be considered certified, \addtextrevision{and the remaining replicas
  are corrupt.}
As there are at least \(\s{}\) honest replicas, honest replicas can
create certified batches without requiring signatures from faulty
replicas.

Unlike in the previous adversarial model, property \PrTermination
still holds under this adversary.
Since the SBC works properly, honest replicas can agree on the
content of each batch, and all valid transaction requests added to
honest replicas eventually appear in an agreed batch.
As there are enough honest replicas to certify batch tags, each batch
agreed by them will be posted as a legal batch tag.
Consequently, all valid transaction requests added to honest replicas
eventually appear in a legal batch tag posted in L1.

As with the \advThree Adversary, if this adversary violates any safety
property, it can be detected using our \<fpms>.
However, detecting violations of property \PrUniqueBatch requires the
action of honest replicas.
Since this adversary controls more replicas than the certification
threshold $\s$, faulty replicas can create certified batch tags
without the signature of any honest replicas.
Therefore, faulty replicas can post a certified batch tags whose batch is
not the one agreed by consensus.
In this case, honest replicas can post a legal batch tag with the same
identifier but different batch content, violating
property~\PrUniqueBatch and exposing that the arranger has been
corrupted.
This violation can be detected simply by observing L1 and the
uniqueness \<fpm> can be used to indicate that arranger replicas must
be replaced.
As a consequence, if the arranger under this adversary does not behave
as under adversary \advOne then its corruption can be exposed.

%%% Local Variables:
%%% TeX-master: "main.tex"
%%% TeX-PDF-mode: t
%%% End:

\section{Related Work}
\label{sec:related-work}

To the best of our knowledge, our work is the first to provide
economic incentives for arrangers replicas to behave honestly, and
fraud-proof mechanisms usable on L1 to demonstrate safety violations
by arrangers, and thus, generate evidence of fraud from faulty
replicas.

Current L2s do not provide any mechanism to guarantee that batches
posted by their arranger are legal or available, because their
arranger is \emph{assumed} to be correct.
Arrangers in Optimistic Rollups and ZK-Rollups can post data that does
not correspond to a compressed batch of valid transaction requests.
Similarly, arrangers of Optimiums and Validiums can post a hash
% that is not signed by enough members of the DAC or \marga{i think
% the L1 contract check signatures in Arbitrum AnyTrust}
whose inverse does not represent a batch of new valid transactions.
The implications of these actions are significant, e.g., malicious STFs
could post illegal new L2 blocks that cannot be disputed by other STFs
as they do not know the transaction requests to play the fraud-proof
over the L2 block.
Hence, malicious arrangers in current L2s can block the L2 or, even
worse, create indisputable incorrect L2 blocks.
By contrast, our \<fpms> prevent such issues under all adversarial
scenarios outlined in Section~\ref{sec:threat-models}, which includes
an adversary that controls all arranger replicas.

One of the \<fpm> we provide is to translate hashes, which has also
been studied in~\cite{boneh}.
However, they assume that the translation can be directly verified in
an L1 smart contract using SNARK proofs, while we provide an
interactive game played over the Merkle tree nodes that only requires
to compute and compare hashes.
%
% \paragraph*{Known Problems with DACs}
% 
%
In other DAC protocols, such as Celestia~\cite{celestia}, clients can
accuse replicas outside of L1 of withholding data, and thus suffer
from the following problem.
They cannot distinguish between the following two situations: (1)
malicious clients falsely accuse replicas of withholding data; and (2)
malicious replicas reveal data only after being accused.
In protocols like~\cite{celestia}, this problem has security
implications depending on the reward received by the accusing client:
\begin{compactenum}[(a)]
\item If the net reward is positive even when the data is revealed,
  rational clients can take advantage by falsely accusing replicas
  just to collect the reward.
\item If the net reward is zero (independently of whether the data is
  revealed or not), denial-of-service attacks are possible because
  malicious clients can issue several accusations without penalty.
\item If the net reward is negative, only altruistic clients would
  accuse malicious stakers.
\end{compactenum}

In our protocol, there is no detection and accusation of withholding
data.
Instead, the data availability \<fp> is used by clients to guarantee
access to the data (or remove stakes).
Furthermore, clients incur on a cost when playing the data
availability \<fp> game, ensuring that it is only used when necessary,
so points (a) and (b) do not apply.
Regarding the third scenario, in our protocol clients always pay to
access the data, either by the offchain protocol or through the more
expensive data availability \<fpm>.
However, if the data is malicious or unavailable, clients receive
significant compensation, and if the data is revealed and it is
correct, clients can generate L2 blocks executing the transaction
requests and be rewarded for it.
Replicas, on the other hand, lose their stake if the data is malicious
and earn nothing if the data is correct and revealed only in a data
availability \<fp> game.

\hypertarget{relatedwork}{}
\addtextrevision{
  Other works that deal with the detection of malicious behavior include
  accountability in consensus~\cite{civit2023easy}, and MPC with identifiable abort~\cite{ishai2014secure}.
  The main difference between our work and MPC with identifiable abort
  is that we not only identify faulty replicas but also
  generate evidence of their faulty behavior that anyone can verify,
  like in accountability in consensus.
  However, an advantage of our work over accountability in consensus
  is that anyone observing the posts done by arranger replicas can
  generate evidence of faulty behavior (when safety properties are
  violated), rather than being restricted to the replicas
  running the protocol. }

%%% Local Variables:
%%% mode: latex
%%% TeX-master: "main"
%%% End:

\section{Conclusions and Future Work}
\label{sec:conclusion}

Layer 2 blockchains aim to improve the scalability of current smart
contract blockchains, offering a much higher throughput without
modifying the programming logic and interaction of blockchains.
An important component of L2s is the arranger, which batches
transaction requests, commits these batches hashed, and translates
hashes into batches.
Current implementations of arrangers in most L2s assume an upper bound
in the number of faulty replicas in order to guarantee correctness.
However, these implementations lack mechanisms to detect when this
trust assumption is violated, and they provide no guarantees in cases
where the assumption does not hold.
As consequence, arrangers have the power to influence the evolution of
the L2.

To address this, we described a collection of incentives for replicas
to behave honestly %when creating, signing and translating batches
and penalties for posting incorrect or unavailable batches based on
fraud-proof mechanisms.
These fraud-proof mechanisms are over properties of concrete
algorithms rather than the execution trace of arbitrary algorithms, as
is the case in RDoC and in the fraud-proofs for used for the
correctness of L2 blocks in Optimistic Rollups and Optimiums.
Using these mechanisms a single honest agent can guarantee that only
valid and available batches of transaction requests are executed in
L2, regardless of the number of faulty replicas in the arranger and
the correctness of the arranger, limiting the power of adversaries.
Furthermore, these fraud-proof games can also be used to prove in L1
violations of safety properties from the arranger, therefore also
acting as deterrent for faulty behavior.
Our incentives and fraud-proof mechanisms are independent from the
concrete implementation of the arranger.

\textit{Artifacts}.  This paper is accompanied by
%the following artifact:
a library of mechanized formal proofs of all our fraud-proof
mechanisms from Section~\ref{sec:incentives:challenges}. developed in
the proof assistant Lean4~\cite{Moura.2021.Lean4}. This includes a
proof that there is a winning strategy for the honest player.

\subsection{Future Work}

One promising direction for future work is to extend our fraud-proof
mechanisms to other properties, such as order
fairness~\cite{kelkar2020orderfairness}.

Another interesting line of research is enhancing the L1 to support
more efficient fraud-proof mechanisms.
Most L1 contracts arbitrating fraud-proof mechanisms first perform
some kind of validation (e.g. verify signatures or compute hashes of
input values and compare them with stored hashes), and then update
their state accordingly.
Interestingly, the validation step executed by these L1 contracts is
typically a \emph{pure computation}: its result depends only on the
inputs and not on the blockchain state.
As consequence, these pure computations could be executed in parallel
with all other transactions and other pure computations as they do not
interfere with their outcome.
To leverage this, we propose extending the semantics of L1 contracts
in blockchains like Ethereum to support function parameters of a
special type, representing the result of pure computations.
Functions with these kind of parameters should be only invoked by
external users initiating the transaction, to ensure that the input to
the pure computation do not depend on the blockchain state.
As these pure computations are part of transactions, their results is
consented just like regular transaction.
Unlike regular transactions, which must be executed sequentially, pure
computations could be executed concurrently with other pure
computations and regular transactions, improving the throughput.
Therefore, the gas consumed by pure computations could be lower than
the gas for transactions, similarly as the cost of posting data in
Ethereum is cheaper when using blobs~\cite{eip4844}.
This approach would not only reduce the cost of existing fraud-proof
games but also allow to simplify more complex ones.
For example, consider the decompress-and-hash \<fpm>, that consists in
bisecting the execution trace of a program that decompresses the data
posted, obtains a list of elements, creates a Merkle tree with the
list of elements as leaves, and finish successfully if the Merkle root
is the hash in a given batch tag.
This \<fp> game could be implement as a one-step \<fpm> that consists
of executing the whole program as a pure computation, because its
computation only depends on the string provided.

\hypertarget{fp:lean:conclusion}{}
Finally,\removetextrevision{the Lean4 library is still under development.
The main claims were proved, but still some transformations between
games need to be formally verified.
In particular, we would like to add the notion of history of committed
DAs to prove Integrity 2, and model the evolution of the system, as
well as the interaction between players.
Additionally,} to completely characterize the system
\addtextrevision{in the Lean4 library} we need to add the notion of
time, possibly as bounded interaction
trees~\cite{Guarded.Interaction.Trees} modeling the effects of agents
interacting.

%%% Local Variables:
%%% TeX-PDF-mode: t
%%% TeX-master: "main.tex"
%%% fill-column : 70
%%% TeX-master: "main"
%%% End:

\begin{acks}
  Funded by project PID2022-140560OB-I00 (DRONAC) funded by MICIU/AEI
  /10.13039/501100011033 and ERDF, EU; grant CEX2024-001471-M/ funded
  by MICIU/AEI/10.13039/501100011033, grant PID2022-138072OB-I00 and
  grant PID2022-142290OB-I00, funded by
  MCIN/AEI/10.13039/501100011033/ FEDER, UE.
\end{acks}

%\vfill
%\clearpage

\bibliographystyle{plain}% the mandatory bibstyle
\bibliography{bibfile}

\vfill
\clearpage
\pagebreak

\appendix
\vfill
\pagebreak
\hypertarget{app:comparison}{}
\addtextrevision{\section{Comparison of L2s on Ethereum With Our Work}}
\label{app:decentralizationL2Eth}

The table below summarizes the types of block and batch validation, the
level of decentralization in the arranger (sequencer) for Layer 2
projects based on Ethereum and the approach proposed in this work.
% %

 \begin{table}[h]
     \centering
       \begin{tabular}{|l|l|l|l|}
        \hline
        \textbf{Project} & \textbf{Block Val.} & \textbf{Arranger}
         & \textbf{Batch Val.}
        \\
         \hline
         Arbitrum One~\cite{ArbitrumNitro} & Fraud-proof & Centralized
         & None\\
         Base~\cite{base} & Fraud-proof & Centralized & None\\
         OP mainnet~\cite{optimism} & Fraud-proof & Centralized & None\\
         Unichain~\cite{unichain} & Fraud-proof & Centralized & None\\
         zkSync Era~\cite{zksyncera} & ZK-proof & Centralized & None \\
         Starknet~\cite{starknet} & ZK-proof & Centralized & None \\
         Linea~\cite{linea} & ZK-proof & Centralized & None \\
         Katana~\cite{katana} & ZK-proof & Centralized & None \\
         Scroll~\cite{scroll} & ZK-proof & Centralized & None \\
         Sophon~\cite{sophon} & ZK-prood & Semi-dec. & None\\
         Ink~\cite{ink} & Fraud-proof & Centralized & None \\
         Abstract~\cite{abstract} & ZK-proof & Centralized & None \\
         Paradex~\cite{paradex} & ZK-proof & Centralized & None \\
         zkSync Lite~\cite{zksynclite} & ZK-proof & Centralized & None\\
         Arbitrum Nova~\cite{ArbitrumNitro} & Fraud-proof &
                                                            Semi-dec.
         & None \\
         Loopring~\cite{loopring} & ZK-proof & Centralized & None \\
         Polygon zkEVM~\cite{polygonzkEVM} & ZK-proof & Centralized & None
         \\
         Zero Network~\cite{zeronetwork} & ZK-proof & Centralized & None
         \\
         Kinto~\cite{kinto} & Fraud-proof & Centralized & None \\
         Lens~\cite{lens} & ZK-proof & Semi-dec. & None \\
         ZKBase~\cite{zkbase} & ZK-proof & Centralized & None \\
         Term Structure~\cite{termstructure} & ZK-proof & Centralized
         & None
         \\
         Phala~\cite{phala} & ZK-proof & Centralized & None \\
         Cartesi~\cite{honeypot} & Fraud-proof & Decentralized & None \\
         LaChain~\cite{lachain} & ZK-proof & Centralized & None \\
         Space And Time~\cite{spaceandtime} & ZK-proof & Centralized & None
         \\
         \textbf{This Work} & Any & Any & Fraud-proof \\
         \hline
    \end{tabular}
 \label{tab:decentralization-layer2eth}
 \caption{Comparison of block and batch validation types, and the
   level of decentralization in the arranger for L2 projects on
   Ethereum, and the approach proposed in this work. Data
   from~\cite{l2beat}}
\end{table}

% * The decentralization status of the arranger for each project X was
% obtained from \url{https://l2beat.com/scaling/projects/X#operator}
% (accessed: 2025-07-13)

%%% Local Variables:
%%% TeX-PDF-mode: t
%%% TeX-master: "../main.tex"
%%% fill-column : 70
%%% End:

\newpage
\vspace{20em}
\section{Payment system}\label{app:htlc}

The following code shows the pseudo-code that implements a payment
system based on a hashed-timelock contract for clients described in
Section~\ref{sec:incentives:translate}.
% \begin{figure}[th!]
% \begin{lstlisting}[language=Solidity,numbers=none,caption = Pseudo-code
%   for contract Payment, label = Payment]
\begin{lstlisting}[language=Solidity,numbers=none]
  contract Payment {
    address owner, beneficiary;
    bytes32 secret;
    uint deadline;

    constructor(address _beneficiary, bytes32 _secret, uint _deadline){
      beneficiary = _beneficiary;
      secret = _secret;
      deadline = _deadline;
    }
    
    function claim(bytes k) public {
      require(msg.sender == beneficiary);
      require(H(k) == secret);
      beneficiary.transfer(balance); 
    }

    function withdraw() public {
      require(msg.sender == owner);
      require(now > deadline);
      owner.transfer(balance);
    }
  }
\end{lstlisting}
% \end{figure}

%%% Local Variables:
%%% TeX-PDF-mode: t
%%% TeX-master: "../main.tex"
%%% fill-column : 70
%%% End:

\newpage
\hypertarget{app:pseudocodes}{}
\addtextrevision{\section{Fraud-Proofs: Pseudocodes}}
\label{app:pseudocodes}

\subsection{Contracts Arbitrating Fraud-Proofs}

We present now the pseudo-code of the L1 smart contract that
arbitrates the \<fpm> described in Section~\ref{sec:fraudproofs}.
These smart contracts are used by the proposer and the challenger
during a challenge game that implements the fraud-proof mechanisms
described in Section~\ref{sec:fraudproofs}.
An honest agent \Att uses the methods of this contract to play
the game described in Section~\ref{sec:fraudproofs}, with a
guarantee that if \Att follows the honest strategy described below \Att will win the game.
Note that none of this contracts has a loop or invokes remotes
contracts, so termination arguments (and gas calculations) can be
easily performed.

For simplicity and clarity we made the following simplification:

\begin{itemize}
\item Contracts are only used for one game. To manage multiple games
  concurrently, all game-specific variables (e.g., player addresses,
  hashes) should be stored in maps from a unique game ID to their
  values, and each function should take a game\_id parameter to access
  the correct game context.

\item Time management was omitted but it works similarly to Arbitrum:
  \begin{itemize}
  \item Each player has an individual clock
  \item The time of the last move is recorded
  \item The active player's clock decreases while it's their turn.
  \item There is a timeout function that anyone can invoke, if the current player runs out of time, it loses the game.
  \end{itemize}
\end{itemize}

% \subsubsection*{Contract Arbitrating Multi-step Membership Fraud-proof
%   mechanism}
\subsection*{Contract Arbitrating Multi-step Membership Fraud-proof
  mechanism}
\begin{lstlisting}[language=Solidity,numbers=none]
contract multi_step_membership {

  address A, B, activePlayer; //  A is the address of the player that initiate the game (proposer), B is the address of the challenger player; activePlayer tracks whose turn it is
  t_hash top, bottom, middle; // Hashes representing the boundaries and midpoint of the current sub-path
  int initial_position, path_length, bottom_level; // initial_position is the leaf position; path_length is the length of the current sub-path, bottom_level is the Merkle level of the bottom hash
  
  // Initialize the game
  function init (element e, int i, t_hash h_e, t_hash h_m, t_hash merkleRoot_hash, int merkleRoot_levels, address _A, address _B) {
    assert(hasStake(_A))
    assert(hasStake(_B))
    if (hash(e) != h_e) {
      setWinner(_B); // B wins if A lied about e's hash
    } else {
      A = _A;
      B = _B;
      top = merkleRoot_hash;
      middle = h_m
      bottom = h_e;
      initial_position = i;
      path_length = merkleRoot_levels - 1;
      bottom_level = 0; // Leaf level starts at 0 (root is at merkleRoot_levels - 1)
      activePlayer = B;
      emit(initMultistepMembership(top,middle,bottom,initial_position, path_length));  
  }

  // B selects the sub-path to challenge (top or bottom)
  function selectSubpath(bool select_bottom) {
    assert(sender == B && activePlayer == B);

    if (select_bottom) {
      // Select bottom half: move top to middle
      top = middle;
      path_length = floor(path_length / 2);
    } else {
      // Select top half: move bottom to middle
      bottom = middle;
      bottom_level += floor(path_length / 2);
      path_length = ceil(path_length / 2);
    }

    activePlayer = A;
    emit(challengedSubpath(top,bottom,path_lenght, bottom_level));
  }

  // A provides new middle hash for the selected sub-path
  function bisectSubpath(t_hash h_m) {
    assert(sender == A && activePlayer == A && path_length >= 2);
    middle = h_m;
    activePlayer = B;
    
    emit(bisectedSubpath(top,bottom,middle,path_lenght, bottom_level));
  }

  // A reveals the sibling hash when sub-path has length 1 (i.e., only two nodes)
  revealSibling(t_hash h_s) {
    assert(sender == A && activePlayer == A && path_length == 1);

    // Determine if bottom is the left or right child of top
    if (bit_at_level(initial_position, bottom_level) == 1) {
      // Sibling is left of bottom
      c = concatenate(h_s, bottom);
    } else {
      // Sibling is right of bottom
      c = concatenate(bottom, h_s);
    }

    // Verify if the parent hash (top) matches the hash of the concatenation
    if (hash(c) == top) {
      setWinner(A);
    } else {
      setWinner(B);
    }
  }
}  
\end{lstlisting}

%%% Local Variables:
%%% TeX-PDF-mode: t
%%% TeX-master: "../main.tex"
%%% fill-column : 70
%%% End:

\vfill

\subsection*{Contract Arbitrating One-step Membership Fraud-proof
  mechanism}
\begin{lstlisting}[language=Solidity,numbers=none]
contract one_step_membership {
  function oneStep(element e, int initial_position, t_hash h_e, t_hash merkleRoot_hash, int merkleRoot_levels, t_hash[] proof, address _A, address _B) {
    assert(hasStake(_A))
    assert(hasStake(_B))
    if (hash(e) != h_e) {
      setWinner(_B); // B wins if A lied about e's hash
    } else {
      // Recompute the Merkle root from the proof
      t_hash h = h_e;
      for(level=0; level<merkleRoot_levels; level++) {
        // Determine the order of concatenation based on the bit at the current level
        if (bit_at_level(i, level) == 1) {
          c = concatenate(proof[level], h);
        } else {
          c = concatenate(h, proof[level]);
        }
        // Hash the concatenated pair to move up one level in the Merkle tree
        h = hash(c)
      }
      // Check if the reconstructed root match the claimed Merkle root to decide the winner
      if (h == merkleRoot_hash){
        setWinner(A);
      } else {
        SetWinner(B);
      }
  }
}  
\end{lstlisting}

%%% Local Variables:
%%% TeX-PDF-mode: t
%%% TeX-master: "../main.tex"
%%% fill-column : 70
%%% End:

\vfill

\subsection*{Contract Arbitrating Certifiability Fraud-proof
  mechanism}
\begin{lstlisting}[language=Solidity,numbers=none]
contract certifiability {
  t_public_key[] arranger_replicas_pk;// Public keys belonging to the arranger replicas

  int S;// Minimum number of signatures required for certification
  
  // Verifies the aggregated batch signature is correct
  function verifyBatchSignature(t_batch_tag bt) {
    assert(hasStake(caller));
    t_signature sig = bt.sig.aggregater_sig;
    int signers_mask = bt.sig.signers_mask;
    t_public_key agg_pk = empty_pk;

    // Aggregate the public keys of all signers according to the mask
    for (i in signers_mask) {
      agg_pk = aggregate_pks(agg_pk, arranger_replicas_pk[i]);
    }

    // Verify the aggregated signature against the aggregated public key and the batch hash
    if (verify_signature(sig, acc_pk, bt.hash)) {
      removeStake(caller);
    } else {
      setWinner(caller);
      removeStake(bt, 'invalid_sig');
    }
  }

  // Check if enough signers participated in the batch signature
  function checkBatchSignatureLength(t_batch_tag bt) {
    assert(hasStake(caller));
    int signers_mask = bt.sig.signers_mask;

    // Verify that the number of signers is at least the required threshold S
    if (maskSize(signers_mask) >= S) {
      removeStake(caller);
    } else {
      setWinner(caller);
      removeStake(bt, 'invalid_sig');
    }
  }
}  
\end{lstlisting}

%%% Local Variables:
%%% TeX-PDF-mode: t
%%% TeX-master: "../main.tex"
%%% fill-column : 70
%%% End:

\vfill
\subsection{Contract Arbitrating Validity Fraud-proof
  mechanism}
\begin{lstlisting}[language=Solidity,numbers=none]
contract validity {
    // Initializes the validity game using a multi-step membership game
  function initMultiStep(t_batch_tag bt, element e, int i, t_hash h_e, t_hash h_m, address A, address B) {
    assert(isStaker(B,bt));
    assert(!valid(e)); //Ensure that the element is not valid transaction request
    multi_step_membership.init(e,i,h_e,h_m,bt.root.hash,
    bt.root.levels,A,B);   
  }

  // Initializes the validity game using a one-step membership game
  function initOneStep(t_batch_tag bt, element e, int i, t_hash h_e, t_hash[] proof, address A, address B) {
    assert(isStaker(B,bt));
    assert(!valid(e))//Ensure that the element is not valid transaction request
    one_step_membership.oneStep(e,i,h_e,bt.root.hash, bt.root.levels, proof);   
  }
  
}   

\end{lstlisting}

%%% Local Variables:
%%% TeX-PDF-mode: t
%%% TeX-master: "../main.tex"
%%% fill-column : 70
%%% End:

\vfill
\subsection*{Contract Arbitrating Integrity1 Fraud-proof
  mechanism}
\begin{lstlisting}[language=Solidity,numbers=none]
contract integrity1 {
  address A;                // Address of the player that initializes the game
  t_batch_tag bt;           // Challenged batch tag
  element e;                // Repeated element as claimed by A
  int[2] i;                 // Indices of element e as claimed by A 

  // Initialize the integrity1
  function init(t_batch_tag _bt, element _e, int[2] _i) {
    bt = _bt;              
    e = _e;                
    i = _i;                
    A = caller;
    emit(initIntegrity1(bt,e,i));
  }

  // Select a path to challenge using one-step membership game by claiming that at one of the positions there is a different element e2.
  function selectPathMembershipOneStep(int pos, element e2, t_hash h_e2, t_hash[] proof) {
    assert(isStaker(caller, bt));
    assert(pos == 0 || pos == 1);
    assert(e2 != e);// Only allow challenge if e2 is different from e
    // Initiate a one-step Merkle membership game
    one_step_membership.init(e2,i[pos],h_e2, proof,bt.root.hash,bt.root.levels,caller,A);
  }

  // Select a path to challenge using multistep-step membership game by claiming that at one of the positions there is a different element e2.
  function selectPathMembershipMultiStep(int pos, element e2, t_hash h_e2,t_hash h_m) {
    assert(isStaker(caller, bt));
    assert(pos == 0 || pos == 1);
    assert(e2 != e);// Only allow challenge if e2 is different from e
    // Initiate a multi-step membership game
    multi_step_membership.init(e2,i[pos],h_e2,h_m,bt.root.hash,bt.root.levels,caller,A);
  }
}

\end{lstlisting}

%%% Local Variables:
%%% TeX-PDF-mode: t
%%% TeX-master: "../main.tex"
%%% fill-column : 70
%%% End:

\vfill
\subsection*{Contract Arbitrating Integrity2 Fraud-proof
  mechanism}
\begin{lstlisting}[language=Solidity,numbers=none]
contract integrity2 {
  address A;                // Address of the player that initializes the game
  t_batch_tag[2] bt;        // Challenged batch tags
  element e;                // Repeated element as claimed by A
  int[2] i;                 // Indices of element e in each batch tag as claimed by A 

  // Initialize the integrity2
  function init(t_batch_tag[2] _bt, element _e, int[2] _i) {
    bt = _bt;              
    e = _e;                
    i = _i;                
    A = caller;
    emit(initIntegrity2(bt,e,i));
  }

  // Select a path to challenge using one-step membership game by claiming that at one of the positions there is a different element e2.
  function selectPathMembershipOneStep(int pos, element e2, t_hash h_e2, t_hash[] proof) {
    assert(isStaker(caller, bt));
    assert(pos == 0 || pos == 1);
    assert(e2 != e);// Only allow challenge if e2 is different from e
    // Initiate a one-step Merkle membership game
    one_step_membership.init(e2,i[pos],h_e2, proof,bt[pos].root.hash,bt[pos].root.levels,caller,A);
  }

  // Select a path to challenge using multistep-step membership game by claiming that at one of the positions there is a different element e2.
  function selectPathMembershipMultiStep(int pos, element e2, t_hash h_e2,t_hash h_m) {
    assert(isStaker(caller, bt));
    assert(pos == 0 || pos == 1);
    assert(e2 != e);// Only allow challenge if e2 is different from e
    // Initiate a multi-step membership game
    multi_step_membership.init(e2,i[pos],h_e2,h_m,bt[pos].root.hash,bt[pos].root.levels,caller,A);
  }
}

\end{lstlisting}

%%% Local Variables:
%%% TeX-PDF-mode: t
%%% TeX-master: "../main.tex"
%%% fill-column : 70
%%% End:

\vfill
\subsection*{Contract Arbitrating Unique Batch Fraud-proof
  mechanism}
\begin{lstlisting}[language=Solidity,numbers=none]
contract unique_batch {
  t_public_key[] arranger_replicas_pk; // Public keys belonging to the arranger replicas
  int S; // Minimum number of signatures required for certification

  // Verifies the two batch tags are not unique
  function uniqueness(t_batch_tag[2] bt) {
    assert(hasStake(caller));
    // Ensure the two batch tags refer to different batches 
    assert(bt[0].id != bt[1].id);
    // Ensure the two batch tags claim the hash
    assert(bt[0].root == bt[1].root);

    //Ensure that both batch tags are certified
    for (i = 0; i < 2; i++) {
      t_signature sig = bt[i].sig.aggregated_sig;
      int signers_mask = bt[i].sig.signers_mask;
      t_public_key agg_pk = empty_pk;
      for (i in signers_mask) {
        agg_pk = aggregate_pks(agg_pk, arranger_replicas_pk[i]);
      }
      assert(verify_signature(sig, agg_pk, bt[i].hash));
      assert(maskSize(signers_mask) >= S);
    }

    removeStake(bt[0], 'unique_batch');
    removeStake(bt[1], 'unique_batch');
    reward(caller, 'unique_batch');
    emit(replaceReplicas(bt[0].sig, bt[1].sig));
  }
}

\end{lstlisting}

%%% Local Variables:
%%% TeX-PDF-mode: t
%%% TeX-master: "../main.tex"
%%% fill-column : 70
%%% End:

\vfill
\subsection*{Contract Arbitrating Data Availability Fraud-proof
  mechanism}
\begin{lstlisting}[language=Solidity,numbers=none]
contract data-availability {
  address A; // A is the address of the player requesting the data
  t_batch_tag bt; //bt is the batch tag for which player A is
  requesting the data
  t_compress data; //claimed compressed version of the batch corresponding to bt

  int S;// Minimum number of signatures required for certification

  
  //Initialize a data-availability game requesting the batch of transaction requests corresponding to batch tag _bt
  function init(t_batch_tag _bt) {
    assert(hasStake(caller));
    bt = _bt;
    A = caller;
    emit(initDataAvailability(bt));
  }

  //Provide claimed compressed version of the batch corresponding to bt
  function postCompressed(t_compress _data) {
    assert(certified(data,S)); //Verifies that the compressed data is certified
    data = _data;
    emit(compressedBatch(bt.root, data));
  }

  // Initiate game that bisects the execution trace of program decompress_and_hash 
  function initDecompressAndHash(t_state init, t_state final, int n,
  t_state middle) {
    decompress_and_hash_game.init(bt.root, data, init,final,n,middle);
  }
  
}  
\end{lstlisting}

%%% Local Variables:
%%% mode: latex
%%% TeX-master: "../main.tex"
%%% End:

\subsection{Honest Strategies}

We present now the strategies of the honest player for games that are
not one-step.
In one-step games, the strategy is trivial: the honest player simply
reveals the required data in the single interaction step, with no
further decisions and the L1 arbitrating contract declares immediate
victory.

For multi-step games, we specify for each player their
\emph{knowledge}---the information available to them--- and the
decision that the player must take at each stage of the game.

\subsection*{Honest strategies for players in the multistep membership
  \<fpm>}

The following pseudocode describes the honest strategies of players \Att
and \Btt in the multistep membership \<fp>
(see Section~\ref{sec:fraudproofs:membership}).

Both players \Att and \Btt know all elements in the Merkle tree
$\<mt>$.

Player \Att initiates the game by claiming that an element $e$ is part
of the Merkle tree $\<mt>$ at a specific position $i$.
During the game, \Att must respond to challenges by revealing nodes
along the Merkle proof of the membership of $e$.

If the claim is incorrect (i.e. \Att is dishonest), an honest player \Btt
challenges the claim by asserting that the element $e$ does not appear
at the claimed position.
Player \Btt examines the hashes provided by \Att and chooses the part
of the path to continue the game: the top subpath is selected when
\Att’s proposed hash does not match the expected hash from the Merkle
tree root; otherwise, the bottom subpath is chosen.
In this way, when \Att claim is incorrect, \Btt keeps the invariant
that the top node in the challenged path is part of the path from the
\(i\)-th leaf to the root of \(mt\), but not the bottom node.

 \begin{algorithm}[H]
  \caption*{\small Player \Att honest strategy}
   \small
   \begin{algorithmic}[1]
     \State  \textbf{Knowledge:}  Merkle Tree $\<mt>$
     \State \textbf{Initial claim:} $(e,i,\<mt>.\<leaves>[i],h_m,\<mt>.\<root>.\<hash>,
     mt.root.levels$) such that\\
     \hspace{2em} $\hash(e) = \<mt>.\<leaves>[i]$\\
     \hspace{2em} $h_m$ is the hash in the node in the middle of the
     path from the $i$-th leaf to the root in $\<mt>$
     \Upon{$\<SubpathChallenged>(\<top>, \<bottom>, \<middle>, \<pathlength>, \<bottomlevel>)$} 
    \If{$\<pathlength> >= 2$} 
      \State $\<middlelevel> = \<bottomlevel> + floor(\<pathlength> / 2);$
      \State $\<indexatlevel>d = \<bitatpos>(i,\<middlelevel>)$
      \State $h\_m = \<mt>.\<getHashAt>(\<middlelevel>, \<indexatlevel>);$
      \State \textbf{invoke} $multi\_step\_membership.\<bisectSubpath>(h_m);$
    \Else 
       \State $\<siblingindex> = \<bitatpos>(i,\<bottomlevel> \oplus 1);$
       \State  $h_s = \<mt>.\<getHashAt>(\<bottomlevel>, \<siblingindex>);$
             \State \textbf{invoke} $multi\_step\_membership.\<revealSibling>(h_s);$
      \EndIf
      \EndUpon
    \end{algorithmic}
  \end{algorithm}

 \begin{algorithm}[H]
  \caption*{\small Player \Btt honest strategy}
   \small
   \begin{algorithmic}[1]
     \State  \textbf{Knowledge:}  Merkle Tree $\<mt>$
     \State \textbf{Disputed claim:} $(e,i,\hash(e),h_m,\<mt>.\<root>.\<hash>,
     mt.\<root>.\<levels>)$\\
     such that\\
     \hspace{2em} $\hash(e) \neq \<mt>.\<leaves>[i]$
     \Function{\<SelectSubpath>}{$\bottomlevel,\pathlength, h_m$}
        \State $\<middlelevel> = \<bottomlevel> + floor(\<pathlength> / 2)$
     \State $\<indexatlevel> = \<bitatpos>(i, \<middlelevel>)$
    \State $\<expectedmiddlehash> = \<mt>.\<getHashAt>(\<middlelevel>, \<indexatlevel>)$
    \If{$\<expectedmiddlehash> ==  h_m$}
     \State \textbf{invoke} $multi\_step\_membership.\<selectSubpath>(\<false>)$
     \Comment{$h_m$ is correct, go up}
    \Else
     \State \textbf{invoke} $multi\_step\_membership.\<selectSubpath>(\<true>)$
     \Comment{$h_m$ is wrong, go down}
   \EndIf
   \EndFunction
   \Upon{$\<initMultistepMembership>(\<mt>.\<root>.\<hash>,h_e,h_m,i,
     mt.\<root>.\<levels>)$} 
    \State $\<selectSubpath>(0, \<pathlength>, h_m)$
  \EndUpon
  \Upon{$\<SubpathBisected>(\<top>, \<bottom>, \<middle>, \<pathlength>, \<bottomlevel>)$}
   \State $\<selectSubpath>(\<bottomlevel>, \<pathlength>, \<middle>)$
  \EndUpon

    \end{algorithmic}
  \end{algorithm}
  
%%% Local Variables:
%%% TeX-PDF-mode: t
%%% TeX-master: "../main.tex"
%%% fill-column : 70
%%% End:

\subsection*{Honest strategies for players in the validity \<fpm>}
The validity \<fpm> is the same as the multistep membership \<fp>,
except that at the beginning, it checks whether the element is a valid
transaction request.
Therefore, the honest strategy for each player is the same as in the
multistep membership \<fp>.

\subsection*{Honest strategies for players in the integrity \<fpms>}

The following pseudocode defines the honest strategies for players
\Att and \Btt in the integrity 1 and integrity 2 \<fps>
(see Section~\ref{sec:incentives:challenges}).

In the integrity 1 \<fp>, both players know all transaction requests
corresponding to a given batch tag $\<bt>$.
Player \Att initiates the game by claiming that an element $e$ appears
in two different positions in $\<bt>$.
If the claim is challenged, the game continues as a multistep
membership game.
In this case, player \Att follows the strategy of \Btt from the
multistep membership game.
If \Att’s claim is incorrect (that is, \Att is dishonest), an honest
player \Btt disputes the claim, asserting that at least one of the
positions contains an element different from $e$.
Then, the game continues as a multistep membership game, with \Btt now
following the strategy of \Att from that game.
 \begin{algorithm}[H]
  \caption*{\small Player \Att honest strategy}
   \small
   \begin{algorithmic}[1]
     \State \textbf{Knowledge:}  Merkle Tree $\<mt>$ such that
     $\<bt>.\<root> = \<mt>.\<root>$
     \State \textbf{Initial claim:} $(\<bt>,e,i)$ such that\\
     \hspace{2em} $\hash(e) = \<mt>.\<leaves>[i[0]]$\\
     \hspace{2em} $\hash(e) = \<mt>.\<leaves>[i[1]]$
     \Upon{$\<initMultistepMembership>(\<mt>.\<root>.\<hash>,h_{e2},h_m,i2,
     mt.\<root>.\<levels>)$ with $e2 \neq e$ and $i2 \in i$}
     \State\textbf{Follow strategy of Player \Btt in membership
       multistep \<fp>.}
     \EndUpon
    \end{algorithmic}
  \end{algorithm}

 \begin{algorithm}[H]
  \caption*{\small Player \Btt honest strategy}
   \small
   \begin{algorithmic}[1]
     \State \textbf{Knowledge:}  Merkle Tree $\<mt>$ such that
     $\<bt>.\<root> = \<mt>.\<root>$
     \State \textbf{Disputed claim:} $(\<bt>,e,i)$ such that\\
     \hspace{2em} $\hash(e) \neq \<mt>.\<leaves>[i[0]]$ or\\
     \hspace{2em} $\hash(e) \neq \<mt>.\<leaves>[i[1]]$
  \Upon{$\<initIntegrity1>(\<bt>,e,i)$}
     \If{$\hash(e) \neq \<mt>.\<leaves>[i[0]]$}  
     \State $\<selectPath>(0,\<mt>.\<leaves>[i[0]],\<mt>.\<getHashInMiddlePath>(i))$
   \Else
     \State $\<selectPath>(1,\<mt>.\<leaves>[i[1]],\<mt>.\<getHashInMiddlePath>(i))$
     \EndIf
  \EndUpon
  \Upon{$\<initMultistepMembership>(...)$}
     \State\textbf{Follow strategy of Player \Att in membership
       multistep \<fp>.}
  \EndUpon

    \end{algorithmic}
  \end{algorithm}
  
%%% Local Variables:
%%% TeX-PDF-mode: t
%%% TeX-master: "../main.tex"
%%% fill-column : 70
%%% End:

The main difference between the integrity 1 \<fp> and integrity 2
\<fp>, is that in later the duplicated element appear in two different
batch tags.
Therefore, the difference for the players is that they must known the
content of both batch tags, but the strategies are essentially the same.

 \begin{algorithm}[H]
  \caption*{\small Player \Att honest strategy}
   \small
   \begin{algorithmic}[1]
     \State \textbf{Knowledge:} Merkle Trees $\<mt>_0$ and $\<mt>_1$ such that
     $\<bt>_0.\<root> = \<mt>_0.\<root>$ $\<bt>_1.\<root> = \<mt>_1.\<root>$
     \State \textbf{Initial claim:} $(\<bt>,e,i)$ such that\\
     \hspace{2em} $\hash(e) = \<mt>_{i[0]}.\<leaves>[i[0]]$\\
     \hspace{2em} $\hash(e) = \<mt>_{i[1]}.\<leaves>[i[1]]$
     \Upon{$\<initMultistepMembership>(\<mt>_j.\<root>.\<hash>,h_{e2},h_m,i2,
     mt.\<root>.\<levels>)$ with $e2 \neq e$ and $i2 = i[j]$}
     \State\textbf{Follow strategy of Player \Btt in membership
       multistep \<fp>.}
     \EndUpon
    \end{algorithmic}
  \end{algorithm}

 \begin{algorithm}[H]
  \caption*{\small Player \Btt honest strategy}
   \small
   \begin{algorithmic}[1]
     \State  \textbf{Knowledge:} Merkle Trees $\<mt>_0$ and $\<mt>_1$ such that
     $\<bt>_0.\<root> = \<mt>_0.\<root>$ $\<bt>_1.\<root> =
     \<mt>_1.\<root>$
     
     \State \textbf{Disputed claim:} $(\<bt>,e,i)$ such that\\
     \hspace{2em} $\hash(e) \neq \<mt>_{i[0]}.\<leaves>[i[0]]$ or\\
     \hspace{2em} $\hash(e) \neq \<mt>_{i[1]}.\<leaves>[i[1]]$
  \Upon{$\<initIntegrity1>(\<bt>,e,i)$}
     \If{$\hash(e) \neq \<mt>_{i[0]}.\<leaves>[i[0]]$}  
     \State $\<selectPath>(0,\<mt>_{i[0]}.\<leaves>[i[0]],\<mt>.\<getHashInMiddlePath>(i))$
   \Else
     \State $\<selectPath>(1,\<mt>_{i[1]}.\<leaves>[i[1]],\<mt>.\<getHashInMiddlePath>(i))$
     \EndIf
  \EndUpon
  \Upon{$\<initMultistepMembership>(...)$}
     \State\textbf{Follow strategy of Player \Att in membership
       multistep \<fp>.}
  \EndUpon

    \end{algorithmic}
  \end{algorithm}
  
%%% Local Variables:
%%% TeX-PDF-mode: t
%%% TeX-master: "../main.tex"
%%% fill-column : 70
%%% End:

\subsection*{Honest strategies for players in the data availability \<fpm>}

The following pseudocode defines the honest strategies of players \Att
and \Btt in the data availability \<fp>
(see Section~\ref{sec:fraudproofs:data-availability}).

Player \Att initiates the data availability \<fp> to learn the batch
of transaction requests associated with a given batch tag $\<bt>$.
If the compressed batch posted by \Btt does not correspond to $\<bt>$,
then \Att continues the interaction through the
\emph{decompress-and-hash \<fpm>}.

An honest player \Btt knows a certified compressed batch $\<data>$
corresponding to the batch tag $\<bt>$.
\Btt’s strategy consists of publishing $\<data>$ when the game is
initiated and following the \emph{decompress-and-hash \<fpm>} strategy
if challenged.

 \begin{algorithm}[H]
  \caption*{\small Player \Att honest strategy}
   \small
   \begin{algorithmic}[1]
     \State \textbf{Knowledge:} $\<bt>$
     \State \textbf{Initial claim:} $\<initDataAvailability>(bt)$
     \Upon{$\<compressedBatch>(\<bt>.\<root>, \<data>)$}
       \If{$\<decompressAndHash>(\<data>) \neq \<bt>.\<root>$}
         \State
         $\<initDecompressAndHash>(P(\<data>,\<bt>.\<root>).\<initState>,$
         \State \hspace{3em} $P(\<data>,\<bt>.\<root>).\<finalState>,P(\<data>,\<bt>.\<root>).\<midState>)$
       \EndIf
      \EndUpon
     \end{algorithmic}
  \end{algorithm}

 \begin{algorithm}[H]
  \caption*{\small Player \Btt honest strategy}
   \small
   \begin{algorithmic}[1]
     \State \textbf{Knowledge:}  Certified compressed batch
     $\<data>$ such that $\<decompressAndHash>(\<data>) = \<bt>.\<root>$
     \State \textbf{Disputed claim:} $(\<bt>)$
  \Upon{$\<initDataAvailability>(\<bt>)$}
     \State $\<postCompressed>(\<data>)$
  \EndUpon
    \end{algorithmic}
  \end{algorithm}

%%% Local Variables:
%%% mode: latex
%%% TeX-master: "../main"
%%% End:

Finally, we do not present the strategy of the
\emph{decompress-and-hash \<fpm>} as it consists of bisecting the
execution trace of a program and it is equivalent to the one presented
in RDoC~\cite{canetti2011practical, canetti2013refereed}

%%% Local Variables:
%%% TeX-PDF-mode: t
%%% TeX-master: "../main.tex"
%%% fill-column : 70
%%% End:

\newpage
\hypertarget{app:figures}{}
\addtextrevision{\section{Fraud Proofs: Figures}}
\label{app:figs}

This appendix includes figures that represent graphically the \<fp>
games introduced in
Section~\ref{sec:fraudproofs}.
We only include figures for \<fps> that involve more than one step.

Positions are represented as ovals and sub-games with squares.
The initial position of the game is \inlinefrugal{init}, marked with a
double oval.
Arrows represent moves, whose destination can be other positions in
the game or other games in which case the initial claim and the role
of challengers and defender is provided.
Arrows with no destination represent disputes that are resolved
immediately by the L1 contract.
Orange arrows represent enabled moves for the \inlinefrugal{proposer}
(player that initiates the game), while red dashed arrows indicate
enabled moves for the \inlinefrugal{challenger}.

\subsection*{Multistep Membership \<fpm>}
\begin{figure}[h!]
  \centering
  \includegraphics[width=0.6\textwidth]{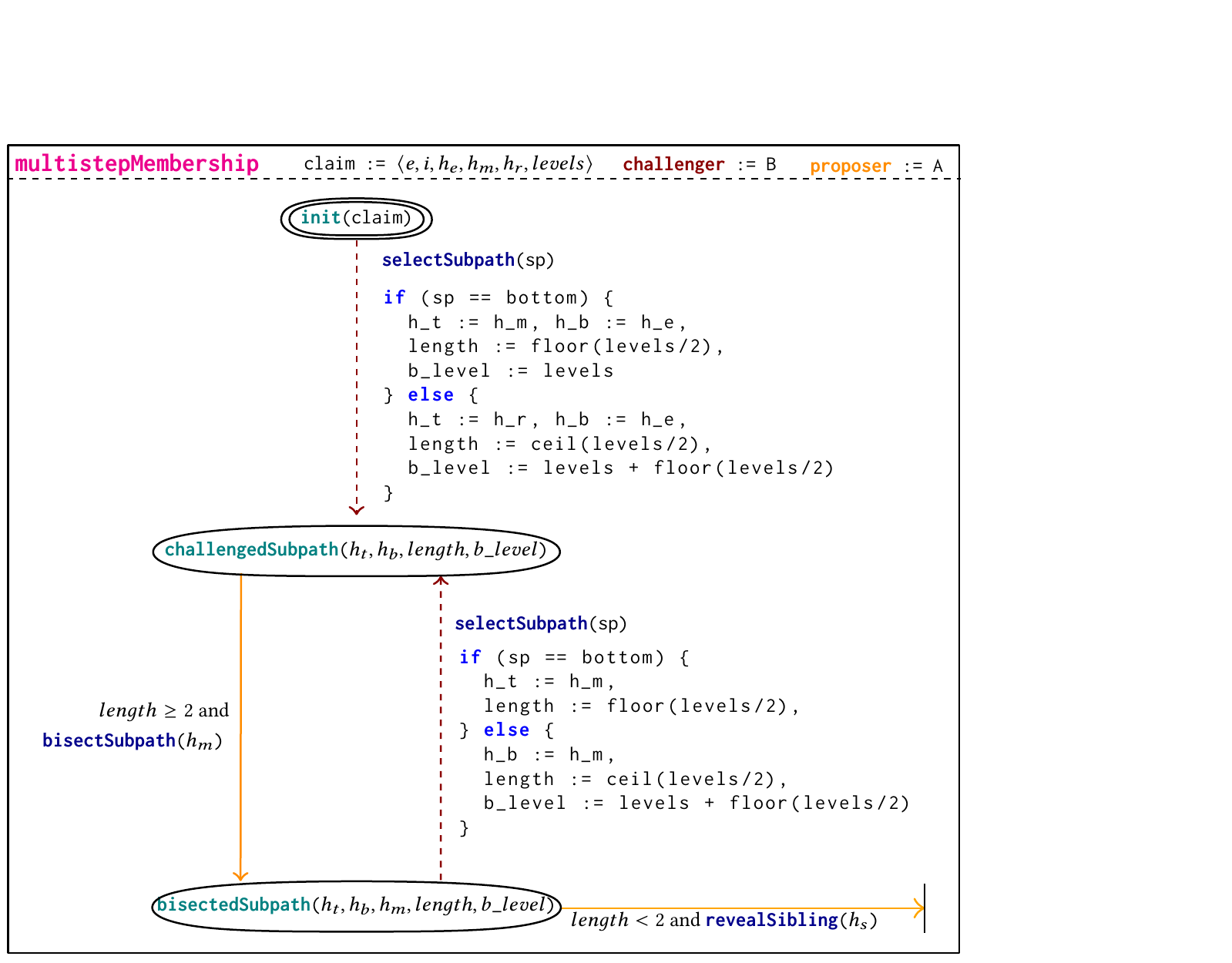}
    \caption{Positions and moves allowed in multistep membership FP.}
    \label{fig:multistepMembership}
\end{figure}

\subsection*{Validity \<fpm>}
The validity \<fpm> is the same as the multistep membership \<fp>,
except that at the beginning, it checks whether the element is a valid
transaction request.
Therefore, its graphical representation is the same as that of the
multistep membership game, as shown in
Fig.~\ref{fig:multistepMembership}.

\subsection*{Integrity 1 \<fpm>}
\begin{figure}[h!]
  \centering
  \includegraphics[width=0.5\textwidth]{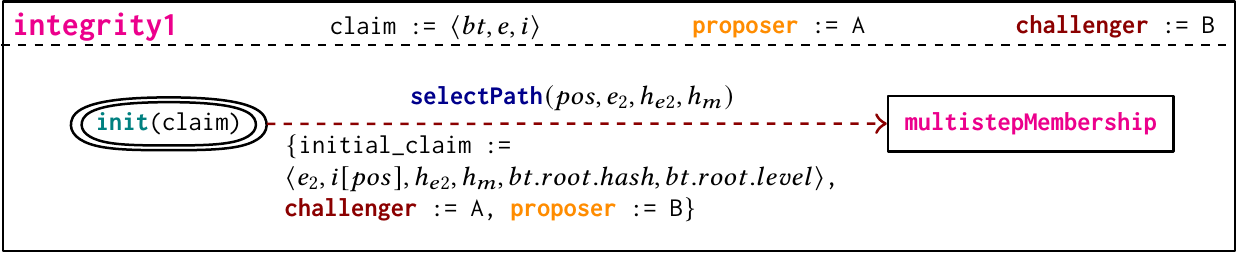}
    \caption{Positions and moves allowed in the integrity 1 FP.}
    \label{fig:integrity1}
\end{figure}

\subsection*{Integrity 2 \<fpm>}
\begin{figure}[h!]
  \centering
  \includegraphics[width=0.5\textwidth]{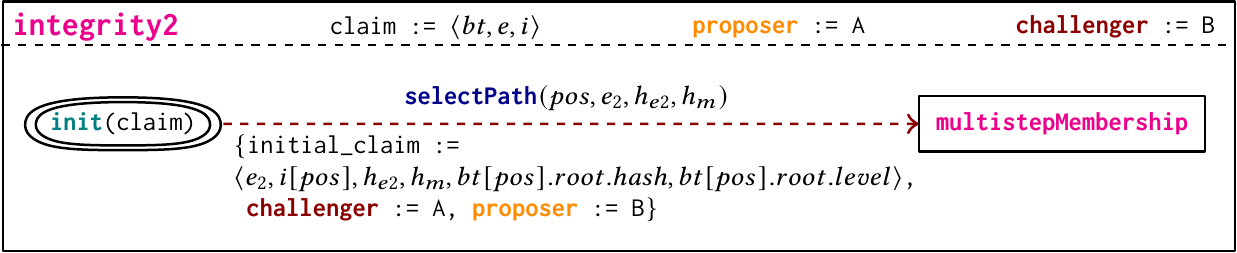}
    \caption{Positions and moves allowed in the integrity 2 FP.}
    \label{fig:integrity2}
\end{figure}

\subsection*{Data Availability \<fpm>}
\begin{figure}[h!]
  \centering
  \includegraphics[width=0.5\textwidth]{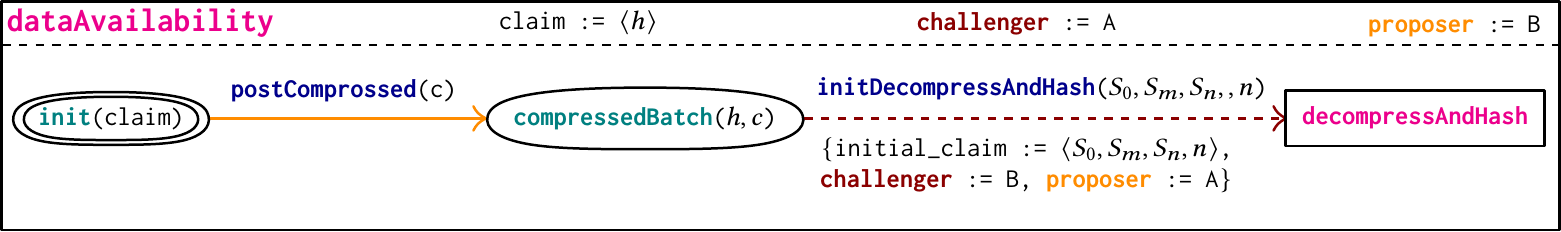}
    \caption{Positions and moves allowed in the data availability FP.}
    \label{fig:da}
\end{figure}

%%% Local Variables:
%%% TeX-PDF-mode: t
%%% TeX-master: "../main.tex"
%%% fill-column : 70
%%% End:

\newpage
\section{Fraud-proof Mechanisms and Incentives - Proofs}%
\label{app:proofs}

We now provide the proof of the propositions presented in
Sections~\ref{sec:fraudproofs} and~\ref{sec:incentives}.

\DAC*
\begin{proof}
Let \(t\) a batch tag posted in L1 that has not consolidate yet, and
\Att an agent that initiates a data availability against \(t\).

There are three possible responses from the agent staking in \(t\):
\begin{enumerate}
\item all staking agents remain silent or try to post data that is
  rejected by the L1 data availability contract. In this case, when
  their time run up, they all lose their stake and batch tag \(t\) is
  discarded,
\item an staking agent post data \(d\) correctly signed but that does not
  correspond to the compression of the batch of transaction requests
  associated with \(t\).
  In this case, when agent \Att try to validate the data will obtain
  one of the following outcomes: (1) the decompression fails, (2) the
  decompressed data does not correspond to a list of elements, or (3)
  the Merkle root does not match the hash in \(t\).
  \Att proceed to the decompress-and-hash \<fpm>, claiming that the
  the arbitrated program with data \(d\) as input fails with the
  outcome obtained will ruining it offchain.
  By corollary~\ref{corollary:rdoc}, agent \Att can win the \<fpm>.
\item an agent posts the compressed version of the batch of transaction
  requests associated to \(t\) correctly signed. In this case, agent
  \Att performs the decompression and validation of the batch
  offchain, and does not proceed to the next \<fpm>. Agent \Att loses
  its stake, but learn the transaction requests corresponding to
  \(t\).
\end{enumerate}

In the first two cases, batch tag \(t\) is discarded.
In the last case, agent \Att (and everyone observing the L1) learn the
transaction requests associated with \(t\).

\end{proof}

\legality*
\begin{proof}
  The proof proceed by cases, depending on the condition violated by
  batch tag \(t\):

  \begin{itemize}
  \item Batch tag \(t\) is not certified. In this case \Att can play
    the certifiability \<fp>. This \<fpm> consists in checking the
    conditions necessary for a batch tag to be considered certified,
    thus \Att wins it.
  \item Batch tag \(t\) contains an invalid element \(e\). \Att knows
    element \(e\) and can initiate the validity \<fp> submitting
    \(e\), which will pass the check, and then move to the membership
    \<fp>.  Since \Att knows all elements in \(t\), by
    proposition~\ref{prop:membership}, \Att can win the membership
    \<fp>.
  \item Batch tag \(t\) contains a duplicated element \(e\). In this
    case, \Att know both positions in which element \(e\) appears in
    the batch corresponding to \(t\), so it can initiate the two
    membership games corresponding to the integrity \<fpm> 1. 
    Again, since, \Att knows all elements in \(t\), by
    proposition~\ref{prop:membership}, \Att can win both membership
    \<fpms>.
  \item Batch tag \(t\) contains an element \(e\) that appears in a
    previous batch tag \(t'\).  In this case, \Att know the positions
    in which element \(e\) appears in the batches corresponding to
    \(t\) and \(t'\), so it can initiate the two membership games
    corresponding to the integrity \<fpm> 2.
    Since \Att knows all elements in \(t\) and \(t'\), by
    proposition~\ref{prop:membership}, \Att can win both membership
    \<fp>.
  \end{itemize}
In all cases, \Att has winning strategy and can discard batch tag
\(t\).
\end{proof}

\uniqueness*
\begin{proof}
  An arranger violates property \PrUniqueBatch by posting two
  certified batch tags, \(t_1\) and \(t_2\), with the same identifier
  but different hash.
  In this case, before the second batch tag consolidates, any agent
  can use the unique batch \<fpm> by submitting \(t_1\) and \(t_2\) to
  the unique batch contract.
  The batches passes the validation check performed by the L1
  contract.
  Consequently, the violation of property \PrUniqueBatch is proved in
  L1, and it is established that the replicas involved should be
  replaced.
\end{proof}

\costs*
\begin{proof}
  We analyze the minimum budget \(B\) based on the conditions violated by
  \(t\):
\begin{compactitem}
\item If \(t\) only violates the certifiability condition, \Att needs
  to play the certifiability \<fp>, which requires a
  $b \geq s_{\KWD{certifiability}} + \CC{\KWD{certifiability}}$.
\item If \(t\) only violates condition validity, then
  \(s_{\KWD{data}} + \CC{\KWD{data}} + s_{\KWD{validity}} +
  \CC{\KWD{validity}}\) tokens are necessary and sufficient.
  \Att needs to know the invalid element in \(t\).
  As all stakers in \(t\) can refuse to participate in the translation
  protocol, \Att might need to play the data \<fp>, which requires
  \(s_{\KWD{data}} + \CC{\KWD{data}}\) tokens.
  Once \Att learned all transactions, \Att needs at least
  \(\s_{\KWD{validity}}+ \CC{\KWD{validity}}\) tokens to play the
  validity \<fp> and make stakers lose their stake.
  Then, in this case, \Att needs at least
  \(s_{\KWD{data}} + \CC{\KWD{data}} + s_{\KWD{validity}} +
  \CC{\KWD{validity}}\) tokens.
  To see that
  \(s_{\KWD{data}} + \CC{\KWD{data}} + s_{\KWD{validity}} +
  \CC{\KWD{validity}}\) tokens are sufficient, \Att can first (1)
  challenge all \(t\) stakers to the data \<fp> until either \Att
  learns all transactions in \(t\) or there no more stakers in \(t\)
  and it is discarded, and then (2) challenge remaining stakers one by
  one to the validity \<fpm>.
  For the first phase, \(s_{\KWD{data}} + \CC{\KWD{data}}\) tokens are
  enough.
  For the second phase, \(s_{\KWD{validity}}+ \CC{\KWD{validity}}\)
  tokens are enough.
  Every time \Att wins a validity \<fp> its budget increases because
  \Att spends \(s_{\KWD{validity}} + \CC{\KWD{validity}}\) tokens but
  receives its stake back and a reward of \(\CR{\KWD{validity}}\)
  tokens, with \(\CR{\KWD{validity}} > \CC{\KWD{validity}}\).
  Then, if \Att starts with at least
  \(s_{\KWD{validity}} + \CC{\KWD{validity}}\) and \Att wins the
  validity \<fp>, \Att will have enough token to continue playing.
  Then, to complete both phases
  \(s_{\KWD{data}} + \CC{\KWD{data}} + s_{\KWD{validity}} +
  \CC{\KWD{validity}}\) tokens are sufficient.
\item If \(t\) only violates \PrIntegrityOne or \PrIntegrityTwo
  condition, \Att must play the integrity \<fp> 1 or 2 but needs to
  learn all transactions in \(t\) before.
  With a reasoning analogous to the previous case, the minimum budget
  needed in this case is
  $
  s_{\KWD{data}} + \CC{\KWD{data}} + s_{\KWD{integrity}} +
  \CC{\KWD{integrity}}$
\item If a batch violates multiple properties, \Att chooses the
  cheapest way to remove it.
\end{compactitem}

Consequently, \Att can remove any illegal or unavailable batch tag
with a budget of at least
\(max(s_{\KWD{certifability}}+\CC{\KWD{certifability}},
 s_{\KWD{data}} + \CC{\KWD{data}} +
  max(s_{\KWD{validity}} + \CC{\KWD{validity}},
  s_{\KWD{integrity1}} + \CC{\KWD{integrity1}},
  s_{\KWD{integrity2}} + \CC{\KWD{integrity2}}))\) L1 tokens.
\end{proof}

\end{document}